\documentclass[final,leqno,onefignum,onetabnum]{siamltex1213}

\usepackage{amssymb,amsmath}
\usepackage{color}
\usepackage{comment}
\usepackage{calrsfs}
\usepackage{amsfonts}
\usepackage{mathrsfs}

\DeclareMathAlphabet{\pazocal}{OMS}{zplm}{m}{n}

\newtheorem{rem}{Remark}

\newcommand{\bb}{\mathbb}
\newcommand{\gs}{\sinh^\diamond}

\newcommand{\R}{{\mathbb{R}}}

\title{Edge modification criteria for enhancing \\ the communicability of digraphs}
\author{Francesca Arrigo\footnotemark[2] \and Michele Benzi\footnotemark[3]}
\begin{document}
\maketitle

\renewcommand{\thefootnote}{\fnsymbol{footnote}}
\footnotetext[2]{Department of Science and High Technology, 
University of Insubria, Como 22100, Italy (\email 
francesca.arrigo@uninsubria.it).}
\footnotetext[3]{Department of Mathematics and Computer Science, 
Emory University, Atlanta, Georgia 30322, USA (\email benzi@mathcs.emory.edu).
The work of this author was supported by National Science Foundation grant
DMS-1418889.} 
\renewcommand{\thefootnote}{\arabic{footnote}}

\begin{abstract}
We introduce new broadcast and receive communicability indices that 
can be used as global measures of how effectively information is
spread in a
directed network. Furthermore,
we describe fast and effective criteria for the selection of edges to be added
to (or deleted from) a given directed network so as to enhance these
network communicability measures. Numerical experiments illustrate the
effectiveness of the proposed techniques.
\end{abstract}



\pagestyle{myheadings}
\thispagestyle{plain}
\markboth{{\sc Francesca Arrigo and Michele Benzi}}{Edge modification criteria for digraphs}


\section{Introduction}
The concept of {\em network communicability}, first introduced  
by Estrada and Hatano in \cite{estradahatano2}, is being increasingly 
recognized as an important metric in the structural analysis of
networks. The communicability between two nodes $i$ and $j$ is defined
as the $(i,j)$th entry in the exponential of the adjacency matrix of the
network (or some scaled version of it). This choice can be justified
on graph-theoretic grounds based on the concept of  walks in a graph, and
also from a statistical physics point of view if we regard a network
as a system of coupled oscillators and consider the
associated thermal Green's function \cite{EHB12}. To date, there have been
a number of applications of communicability to the analysis of
real-world complex networks, a few of which are surveyed in \cite{EHB12}. 

In \cite{BK13} a new node centrality measure was introduced based on
the notion of {\em total communicability}, which measures how easily a given
node communicates with all the nodes in the network. As pointed out in
\cite{BK13}, this centrality measure is closely related to the notion
of {\em subgraph centrality} \cite{estradarodriguez05}, while being 
much easier to compute in the case of large networks.
In \cite{BK13} it was also proposed to use the sum of all the
total node communicabilities, possibly normalized by the number of nodes, 
as a {\em global} measure of how effective the network is at propagating
information among its nodes. This global index, referred to as
{\em total network communicability}, was further shown in 
\cite{AB14} to provide a good measure of the connectivity and robustness
of complex networks, while being much faster to compute than existing metrics
(such as the closely related {\em free energy} \cite{Estrada2007} and
{\em natural connectivity} \cite{Wu2010,Wu2012}). Indeed, the cost of
estimating the total network communicability using Lanczos-type
methods scales linearly in the
size of the graph for many types of networks \cite{BK13}.

Given that high communicability is often a highly desirable feature
(especially in the case of certain infrastructure, information, and
social networks) it is then natural to ask whether it is possible to
design networks which are at the same time highly sparse (in the sense 
of low average degree) and yet have high total communicability. (This
problem is analogous to that of constructing good expander graphs,
see \cite{HLW06}.) In \cite{AB14} we considered the problem of modifying
an existing sparse network so as to cause the total network communicability
to change in some desired way. The modification can be the
addition of a missing edge, the
deletion of an existing edge, or the 
rewiring of an existing edge. The goal could
be to increase the total communicability of the network as much as possible
(or nearly so), or to sparsify the network while minimizing the drop in
the value of the total communicability, subject to constraints on the
number of edge modifications allowed. In \cite{AB14}, several fast and
effective heuristics have been developed that achieve the desired goal.

A serious limitation of the notion of total communicability is that it
is not well suited to deal with directed networks, and indeed all the 
above mentioned papers deal exclusively with undirected networks.
The main reason is that in a directed graph each node plays two roles,
that of broadcaster and that of receiver of information. It is clear that
a single index cannot discriminate between these two forms of 
communication. In this paper, building in part on the ideas in \cite{BEK13},
we define two new measures of total network communicability, which quantify
how easily information is propagated on a given directed network when the 
two fundamental modes of communication
(broadcasting and receiving) have both to be accounted for.
Furthermore, we generalize the edge modification criteria in 
\cite{AB14} from the undirected
to the directed case, using the newly introduced communicability indices as
the objective functions.

Examples of real-world directed networks
include various information and citation networks, such as corpora of documents
linked to each other by directed edges, for examples hyperlinks between web pages or
in-line references between Wikipedia entries. In such networks, one may want to delete
edges that contribute little to the overall authoritativeness of the network.
In other cases, one may want to add edges from a hub to other nodes so as to
increase the overall efficiency of the network in leading to authoritative documents.
It is therefore of interest to introduce criteria for edge selection aimed at 
(approximately) optimizing  the global communicability properties of directed networks.

A few other authors have previously considered heuristics for edge manipulation
in directed networks.  In \cite{ZLZ}, edge modification criteria are introduced for tuning the
{\em synchronizability} of a network, a property of interest in many settings.
In \cite{Gates}, the authors have considered the potential impact of edge modification 
on epidemic dynamics on contact networks. It is quite possible that our edge
selection criteria may find application in these contexts.

The remainder of the paper is organized as follows. 
Section 2 contains some background notions about digraphs, the 
singular value decomposition,
and network centrality measures.
The new total communicability indices for digraphs are introduced in
section 3. The edge updating/downdating problem is described in section 4.
In section 5 we introduce the proposed heuristics for edge manipulation,
and in section 6 we discuss the result of numerical tests (including timings)
using four real-world directed networks.
Conclusive remarks are found in section 7.


\section{Background}\label{sec:background}
In this section we recall a few definitions and notations associated with 
graphs. 
Let $\pazocal{G}=(\pazocal{V},\pazocal{E})$ be a graph with $n=|\pazocal{V}|$ nodes and 
$m=|\pazocal{E}|$ edges (or links). 
If for all $i,j\in \pazocal{V}$ such that $(i,j)\in \pazocal{E}$ then also $(j,i)\in 
\pazocal{E}$, the graph is 
said to be {\it undirected}, as its edges can be traversed without 
following any prescribed ``direction.''  
On the other hand, if this condition does not hold, namely if there exists $(i,j)\in 
\pazocal{E}$ such that $(j,i)\not\in \pazocal{E}$, then 
the network is said to be {\it directed}. 
A directed graph is commonly referred to as a {\it digraph}. 
If $(i,j)\in \pazocal{E}$ in a digraph, we will write $i\rightarrow j$. 
As for the undirected case, an unweighted digraph can be represented by means of a binary matrix 
$A\in\bb{R}^{n\times n}$ 
whose entries $(A)_{ij}=a_{ij}$ 
are nonzero if and only if  $(i,j)\in \pazocal{E}$. 
An ordered pair $(i,j)\not\in \pazocal{E}$ will be called a {\it virtual edge}. 

Every node $i\in \pazocal{V}$ in a digraph has two types of degree, namely the {\it in-degree} and the 
{\it out-degree}; the first, denoted by $d_{in}(i)$, counts the number of edges of the form $*\rightarrow i$, i.e., the 
number of nodes in $\pazocal G$ from which it is possible to reach $i$ in one step. 
The {\it out-degree}, on the other hand, counts the number of nodes that can be reached from $i$ 
in one step, i.e., the number of edges of the form $i\rightarrow *$, and 
is denoted by $d_{out}(i)$.  
The degrees of a node $i$ can be 
computed as the $i$th entries of the following two vectors:
$$\left\{
\begin{array}{l}
{\bf d}_{out}=A{\bf 1}, \\
{\bf d}_{in}=A^T{\bf 1}=({\bf 1}^TA)^T.
\end{array}
\right. $$
Here ${\bf 1}$ is the vector of all ones and the superscript ``$T$'' 
denotes transposition. 

A {\it walk} of length $k$ is a sequence of (possibly repeated) 
nodes $i_1,i_2,\ldots,i_{k+1}$ such that $i_l\rightarrow i_{l+1}$ for all 
$l=1,\ldots,k$; a walk is said to be {\it closed} if $i_1=i_{k+1}$. 
A {\it path} is a walk with no repeated nodes. 
A digraph is said to be {\it strongly connected} if every two nodes in the network are 
connected through a path of finite length, while it is said to be {\it weakly connected} if 
this property holds when the directionality of the links is disregarded. 

Unless otherwise stated, every digraph in this paper is {\em simple},
i.e., unweighted, weakly connected, and without self-loops or multi-edges.

Let $A=U\Sigma V^T$ be a {\it singular value decomposition} (SVD) of the 
adjacency matrix $A$  \cite{HJ}. 
The matrix $\Sigma\in\bb{R}^{n\times n}$ is diagonal and its 
diagonal entries $(\Sigma)_{ii}=\sigma_i$ are the {\it singular values} of $A$. 
These elements are non-negative and ordered as 
$$\sigma_1\geq\sigma_2\geq\cdots\geq\sigma_r>\sigma_{r+1}=\cdots=\sigma_n=0,$$ 
where $r=\rank(A)$ is the rank of $A$. 
The matrices $U,V\in\bb{R}^{n\times n}$ are 
orthogonal and $U=[{\bf u}_1,{\bf u}_2,\ldots,{\bf u}_n]$ contains the {\it left singular vectors} of $A$, while 
 $V=[{\bf v}_1,{\bf v}_2,\ldots,{\bf v}_n]$ contains the {\it right singular vectors}. 
As is well known, $\Sigma$ is uniquely determined by $A$ but $U$ and $V$ are not.
Given an SVD of $A$, the corresponding
\emph{compact singular value decomposition} (CSVD) of the matrix $A$ is 
given by $A=U_r\Sigma_rV_r^T$, where 
$U_r=[{\bf u}_1,{\bf u}_2,\ldots,{\bf u}_r]\in\bb{R}^{n\times r}$ and 
$V_r=[{\bf v}_1,{\bf v}_2,\ldots,{\bf v}_r]\in\bb{R}^{n\times r}$ consist 
of the first $r$ columns of $U$ and $V$, respectively, 
and $\Sigma_r={\rm diag}(\sigma_1,\sigma_2,\ldots,\sigma_r)\in\bb{R}^{r\times r}$ 
corresponds to the leading $r\times r$ diagonal block 
of $\Sigma$.

\subsection{Hubs and authorities}
Let us briefly recall here a few definitions concerning the dual role every node plays in a digraph. 
In \cite{K99} Kleinberg stated that in directed networks there exist two types of important nodes: 
{\it hubs} and {\it authorities}. 
In particular, each node can be assigned a hub score and an authority score, which quantify its 
ability of playing these two roles. 
Good hubs are those nodes which better broadcast information, while good authorities are those which better receive 
information.
These two types of importance for nodes 
are strongly related through a recursive definition: the importance of a node as hub 
is proportional to the importance as authorities of the nodes it points to. 
Similarly, the importance of a node 
as authority depends on the importance as hubs of the nodes that point to it. 
This recursive definition is highlighted in the implementation of the HITS algorithm (see \cite{K99}), which makes use of the 
eigenvectors corresponding to the leading eigenvalue of the symmetric matrices $AA^T$ (the {\it hub matrix}) and $A^TA$ 
(the {\it authority matrix}) to rank the nodes as hubs and authorities, 
respectively.\footnote{For simplicity, unless otherwise specified,
in this and the next section we assume that the dominant eigenvalue
of $AA^T$ (and therefore of $A^TA$) is simple. This ensures the uniqueness (up
to scalar multiples) of the principal eigenvectors of these matrices, and
therefore of the hub and authority rankings. We refer the reader to \cite{EIHITS}
and to \cite[page 120]{Pagerank} for a discussion of this issue.}

Using the SVD or the CSVD of the adjacency matrix, 
it easily follows that $AA^T=U\Sigma^2U^T = U_r \Sigma_r^2U_r^T$ and $A^TA=V\Sigma^2V^T = V_r\Sigma_r^2V_r^T$. 
Therefore, the vector containing the hub scores is ${\bf u}_1$ while the vector containing the authority scores is ${\bf v}_1$.  
By the Perron--Frobenius theorem \cite{HJ}, from the non-negativity 
and irreducibility
of the hub and authority matrices it follows  
that these principal eigenvectors can be chosen so as to have positive components.  
Hence, ${\bf u}_1> 0$ will be called the {\it hub vector} and the vector 
${\bf v}_1> 0$ will be called the {\it authority  vector}. 

The powers of the hub and authority matrices are related to the number of particular types of walks in the 
digraph. 
Following \cite{BEK13,CEHT10}, we define an {\it alternating walk of length $k$ 
starting with an out-edge} 
as a list of nodes $i_1,i_2,\ldots, i_{k+1}$ such that there exists an edge $(i_l,i_{l+1})$ 
if $l$ is odd and an edge $(i_{l+1},i_{l})$ if $l$ is even. Hence, 
an alternating walk starting with an out-edge has the form
$$i_1\longrightarrow i_2\longleftarrow i_3\longrightarrow\ldots.$$
Similarly, an {\it alternating walk of length $k$ starting with an in-edge} 
is a list of nodes 
$i_1,i_2,\ldots, i_{k+1}$ such that
$$i_1\longleftarrow i_2\longrightarrow i_3\longleftarrow\ldots,$$
 i.e., such that there exists an edge $(i_l,i_{l+1})$ if $l$ is 
even and an edge $(i_{l+1},i_l)$ otherwise.

It is well known that the entries of powers of the adjacency matrix of a
graph can be used to 
count the number of walks of a certain length in the network. 
Similarly, it is known (see, e.g., \cite{CEHT10})
that $[AA^TA \ldots]_{ij}$ (where there are $k$ matrices being multiplied) 
counts the number of alternating walks of length $k$, starting with an out-edge, 
from node $i$ to node $j$,
whereas  $[A^TAA^T\ldots]_{ij}$ (where there are $k$ matrices being multiplied) 
counts the number of alternating walks of length $k$, starting with an in-edge, 
from node $i$ to node $j$.  Thus, $[(AA^T)^k]_{ij}$ and $[(A^TA)^k]_{ij}$ count 
the number of alternating walks of length $2k$.

In the next section, we will show how to use these quantities to define two global 
measures of how effectively the nodes in a digraph exchange information. 


\section{Total network communicabilities for digraphs}\label{sec:TC}
In \cite{BK13} a global measure of how easily information is diffused
across an (undirected) network has 
been defined in terms of the matrix exponential of the adjacency matrix. 
More in detail, recalling that the entries of the matrix exponential count the total number of walks of any 
length between two nodes weighting walks of length $k$ by a factor $\frac{1}{k!}$, the 
{\it total network communicability} has been defined as the sum of all the entries of 
this matrix: 
\begin{equation}\label{tc_undir}
TC(A):= \sum_{i=1}^n\sum_{j=1}^n \left [e^A\right ]_{ij} = {\bf 1}^T e^A {\bf 1},
\quad e^A = \sum_{k=0}^\infty \frac{1}{k!}A^k. 
\end{equation}
This quantity, possibly normalized by $n$, has been empirically shown to provide 
a good measure of how effectively the information flows 
along the network and of how well connected an undirected network is (see \cite{AB14,BK13}).

Let now $A$ be the adjacency matrix of a directed graph.
In analogy with the undirected case, we can
consider the 
total network communicability (\ref{tc_undir}).
In principle, this quantity (possibly normalized by $n$)
gives us an idea of how efficient the network
is, globally, at diffusing information. 
However, by following this approach we would be completely disregarding the 
twofold nature of nodes, which is one of the main features of digraphs.  

To better capture the dual behavior of nodes, we introduce two new global 
indices of communicability 
defined in terms of functions 
of the hub and authority matrices.

\begin{definition}\label{def:DItc}
Let $A$ be the adjacency matrix of a simple digraph and let
 $f:\R\longrightarrow\R$ be a function 
defined on the spectrum of $AA^T$.
The {\rm total hub $f$-communicability} of the digraph is defined as 
$$T_hC(A,f):={\bf 1}^Tf(AA^T){\bf 1}=\sum_{i=1}^nf(\sigma_i^2)({\bf 1}^T{\bf u}_i)^2.$$
Similarly, the {\rm total authority $f$-communicability} of the digraph is defined as
$$T_aC(A,f):={\bf 1}^Tf(A^TA){\bf 1}=\sum_{i=1}^nf(\sigma_i^2)({\bf 1}^T{\bf v}_i)^2.$$
\end{definition}

The motivation for using these quadratic forms as total communicability indices 
is that they 
exploit the recursive definition that relates hubs and authorities in a directed network. 
Assume that the function $f$ can be expressed as a power series of the form
\begin{equation}\label{eq:PS}
 f(t)=\sum_{k=0}^\infty  c_k t^k,\quad c_k\geq 0 \quad \forall k=0,1,\ldots
\end{equation}
Then, an easy computation shows that the total hub $f$-communicability, 
can be described in 
terms of the in-degree vector and of the authority matrix as  
$$T_hC(A,f)=c_0n+c_1\|{\bf d}_{in}\|_2^2+\sum_{k=1}^\infty c_{k+1}{\bf d}_{in}^T(A^TA)^k{\bf d}_{in},$$
thus highlighting the fact that the overall ability of nodes to broadcast information depends on 
their ability of receiving it. 
Note that due to the nonnegativity assumption on the
coefficients in (\ref{eq:PS}), $T_hC(A,f)$ is an inherently nonnegative quantity.

Analogous computations carried out on the total authority $f$-communicability show that this index  
can be completely described in terms of the out-degree vector and of the hub matrix:
$$T_aC(A,f)=c_0n+c_1\|{\bf d}_{out}\|_2^2+\sum_{k=1}^\infty c_{k+1}{\bf d}_{out}^T(AA^T)^k{\bf d}_{out},$$
thus showing that the overall ability of nodes to receive information depends on 
how well they are able to broadcast it. 
Note that $T_aC(A,f)$ is, again, 
always nonnegative.

\begin{rem}
{\rm We stress that
the total hub and authority $f$-communicabilities are invariant under graph isomorphism.
Indeed,
let $\pazocal{G}_1$ and $\pazocal{G}_2$ be two isomorphic graphs with associated adjacency matrices $A_1$ and $A_2$. 
Then there exists a permutation matrix $P$ such that $A_2 = P A_1 P^T$. 
Therefore,
\begin{align*}
T_hC(A_2,f) &= {\bf 1}^Tf(A_2A_2^T){\bf 1} = {\bf 1}^T f(P A_1 P^T P A_1^T P^T) {\bf 1} \\
	    &= {\bf 1}^T P f(A_1 A_1^T)P^T{\bf 1} = {\bf 1}^T f(A_1 A_1^T){\bf 1} = T_hC(A_1,f).
\end{align*}
Similarly,
\begin{align*}
T_aC(A_2,f) &= {\bf 1}^Tf(A_2^T A_2){\bf 1} = {\bf 1}^T P f(A_1^T A_1)P^T{\bf 1} \\
	    &= {\bf 1}^T f(A_1^T A_1){\bf 1} = T_aC(A_1,f).
\end{align*}
}
\end{rem}

In this paper we will focus on the total hub and  authority  $f$-communicabilities 
when the 
function $f(t)=\cosh(\sqrt{t})$ is used in the definition. 
The choice of the function $f(t)$ may seem unusual; however, we argue
that this choice is the most natural one if one wants to ``translate'' the
idea of  total communicability to the case of digraphs. 
Indeed, in the undirected case the total communicability was defined as the sum of all the entries of the matrix exponential. 
This index counts all the walks of any length taking place in the network, 
weighting walks of length $k$ by a factor $\frac{1}{k!}$. 
In the case of a digraph, we need to count all the alternating walks,  
again penalizing longer walks. This is accomplished by taking $f(t) =\cosh (\sqrt{t})$;  
for this choice of $f$ we obtain the {\it total hub communicability} 
as 
\begin{align*}
T_hC(A) &:={\bf 1}^T\left(\sum_{k=0}^\infty \frac{(AA^T)^k}{(2k)!}\right){\bf 1}
	={\bf 1}^T\left(\sum_{k=0}^\infty \frac{(\sqrt{AA^T})^{2k}}{(2k)!}\right){\bf 1}\\
	& ={\bf 1}^T\cosh(\sqrt{AA^T}){\bf 1} = T_hC(A,\cosh(\sqrt{t})),
\end{align*}
and, similarly, the {\it total authority communicability} as  

\begin{align*}
T_aC(A) &:={\bf 1}^T\left(\sum_{k=0}^\infty \frac{(A^TA)^k}{(2k)!}\right){\bf 1}
	={\bf 1}^T\left(\sum_{k=0}^\infty \frac{(\sqrt{A^TA})^{2k}}{(2k)!}\right){\bf 1} \\
	&={\bf 1}^T\cosh(\sqrt{A^TA}){\bf 1} = T_aC(A,\cosh(\sqrt{t})).
\end{align*}

A further justification for the choice of the function $f(t)$ comes 
from considering the following construction (see \cite{BEK13}).
Let 
\begin{equation}\label{eq:A_bipartite}
\mathscr{A}=\left(
\begin{array}{cc}
0 & A \\
A^T & 0
\end{array}
\right)
\end{equation}
be the adjacency matrix of the bipartite graph $\mathscr{G}=(\mathscr{V},\mathscr{E})$ 
obtained from the original digraph represented by $A$. 
This graph has $2n$ nodes forming the set $\mathscr{V}=\pazocal{V}\cup \pazocal{V}'$, 
where $\pazocal{V}$ is the original 
set of nodes and $\pazocal{V}'=\{1'=n+1, 2'=n+2,\ldots,n'=2n\}$ is 
a set of copies of the nodes in $\pazocal{G}=(\pazocal{V},\pazocal{E})$. 
The edges between the elements in $\mathscr{V}$ are undirected and 
$(i,j')\in\mathscr{E}$ with $i\in \pazocal{V}$ 
and $j'\in \pazocal{V}'$ if and only if $(i,j)\in \pazocal{E}$ in the original digraph.

Note that in the bipartite graph the first $n$ nodes contained in $\pazocal{V}$ can be seen as the original nodes of the digraph when they play their 
role of broadcasters of information, while the $n$ copies contained in $\pazocal{V}'$ represent the original nodes in their 
role of receivers. 
It is worth mentioning that the eigenvector of $\mathscr{A}$ corresponding to 
the leading eigenvalue $\lambda_1(\mathscr{A})=\sigma_1$ 
is the vector ${\bf q}_1= \left [\begin{array}{c} {\bf u}_1\\ {\bf v}_1 
\end{array} \right ]$.
The choice of $f(t)=\cosh(\sqrt{t})$ follows from the next result. 

\begin{proposition}{\rm \cite[Proposition 1]{BEK13}}
Let $\mathscr{A}$ be as in \eqref{eq:A_bipartite} and
let $A=U\Sigma V^T$ be an SVD of $A$. Then
\begin{equation}\label{eq:expmA_bipartite}
e^{\mathscr{A}}=\left(
\begin{array}{cc}
\cosh(\sqrt{AA^T}) & U\sinh(\Sigma)V^T\\
V\sinh(\Sigma)U^T & \cosh(\sqrt{A^TA})
\end{array}
\right).
\end{equation}
\end{proposition}

An important feature of this matrix is that its entries are nonnegative. Thus, these quantities can be used 
to describe the importance of nodes and how well they communicate when they are acting as broadcasters or receivers of information 
in the graph~\cite{BEK13}. 
Indeed, the entries of the two diagonal blocks $\cosh(\sqrt{AA^T})$ 
and $\cosh(\sqrt{A^TA})$ provide centrality and 
communicability indices for nodes and pairs of nodes 
when they are all seen as playing the same role in the network. 
More in detail, the diagonal entries of the first diagonal block give the centralities for the nodes in the original 
network when they are seen as broadcasters of information (hubs). 
Likewise, the diagonal of the second block contains the centralities for 
the nodes in their role of receivers (authorities). 
Similarly to the off-diagonal entries of the matrix exponential of an undirected graphs, the off-diagonal entries 
of these diagonal blocks measure how well two nodes, both acting as 
broadcasters (resp., receivers), exchange information. 

As for the off-diagonal blocks in~\eqref{eq:expmA_bipartite}, 
they contain information concerning how nodes exchange information 
when one node is playing the role of broadcaster (resp., receiver) and the other 
is acting as a receiver (resp., broadcaster). 

Thus, the total hub communicability and total authority communicability defined as $T_hC(A)={\bf 1}^T\cosh(\sqrt{AA^T}){\bf 1}$ and 
$T_aC(A)={\bf 1}^T\cosh(\sqrt{A^TA}){\bf 1}$, respectively, 
account for the overall ability of the network of exchanging information when all its nodes are playing the same role of broadcasters 
($T_hC(A)$) 
or receivers ($T_aC(A)$).

\section{Edge modification strategies}\label{sec:strategies}

The main goal of this work is to develop heuristics that can be used 
to add/remove edges from a digraph in order to tune 
the total hub and/or authority communicability.
In particular, we will call {\it update} of $(i,j)\not\in \pazocal{E}$ the addition of this virtual edge to the
network; we want to perform this operation in such a way that this addition 
increases as much as possible the quantities of interest. Note that, 
due to the nonnegativity condition in (\ref{eq:PS}),
the addition of an edge can only increase the total communicabilities $T_hC(A)$ and $T_aC(A)$. 

The operation of removing an edge from the network will be referred 
to as the {\it downdate} 
of an edge. 
Our aim is to select the edge to be removed in such a way that 
the target functions $T_hC$ and $T_aC$ are not penalized too much, 
i.e., their values do not 
drop significantly
as edges are removed.\footnote{Clearly, our approach can be adapted so as to obtain 
the opposite effect if so desired. Indeed, we can adapt our algorithms to select 
edges whose removal heavily penalizes the 
target functions.} 
Both these operations can be described as rank-one modifications of the adjacency 
matrix $A$ of the digraph $\pazocal G$ or, equivalently, as rank-two modifications of the 
adjacency matrix $\mathscr{A}$ of the associated bipartite graph $\mathscr{G}$. 

We first introduce some edge centrality measures that can be used to 
rank the (virtual) edges in the digraph; we 
then use the derived rankings to select which modifications to perform. 
More in detail, a virtual edge having a large centrality is considered 
important and thus its addition 
is expected to highly enhance the total communicabilities. 
On the other hand, we will remove edges that have a low ranking, 
since they are not expected to carry a lot of information; 
thus, their removal is not expected to heavily penalize the hub and 
authorities communicabilities of the network.

The resulting updating and downdating strategies will be similar in spirit 
to those adopted in the undirected case \cite{AB14}.
However, as explained in more detail in the next subsection, we cannot simply
apply the heuristics in \cite{AB14} to the bipartite graph $\mathscr{G}$,
since doing so could lead to possible loss of structure.

\subsection{Bipartite graphs vs.~digraphs} \label{sec:vs}
In this section we will describe two different ways of tackling the problem of selecting $K$ edge modifications to be 
performed on the network in order to tune the communicability indices $T_hC(A)$ and  
$T_aC(A)$.

First we describe how to rank the edges. 
A priori, there are two natural approaches. 
Indeed, given the definitions of communicabilities in terms of the 
function $f(t)=\cosh(\sqrt{t})$, we can 
either work on the matrix $\mathscr{A}$ or on the original adjacency matrix $A$.
In the first case, we would adapt to the matrix $\mathscr{A}$ the 
techniques developed for the
undirected case which performed best according to the results in \cite{AB14}, taking
into account the need to preserve the zero-nonzero block structure of $\mathscr{A}$.
The second approach, on the other hand, requires the introduction of new 
edge centrality measures specially developed for the directed case.

We will show that the new edge centrality measures for digraphs allow us to 
develop heuristics that perform as well as or better than the techniques for undirected 
graphs applied to $\mathscr{A}$. 

We want to stress here that the set of (virtual) edges among which we 
select the modifications is the same in 
both cases, since one wants to preserve the antidiagonal block structure 
(\ref{eq:A_bipartite})
of $\mathscr{A}$. 
Indeed, if a new edge were to destroy the structure, it could not be 
``translated'' into a new directed 
edge for the original digraph.


\subsection{Edge centralities: undirected case}\label{sec:und} 
In the following, we will briefly recall the edge centrality measures 
introduced in \cite{AB14} 
that showed the best performance. 
These will be used on $\mathscr{A}$ to tackle the updating and downdating problems. 

Let $M$ be the adjacency matrix of a simple, undirected graph. 
We call the {\it edge eigenvector centrality} of the (virtual) edge $(i,j)$ the quantity:
$$^e{\pazocal EC}(i,j)=q_1(i)q_1(j),$$
where $q_1(i)$ is the $i$th entry of the Perron vector ${\bf q}_1$ of the matrix $M$ 
(see \cite{HJ,B87}). 
We call {\it edge total communicability centrality} of $(i,j)$ the quantity:
$$^e{\pazocal TC}(i,j)=(e^M{\bf 1})_i(e^M{\bf 1})_j.$$

Let $\lambda_1> \lambda_2\ge \cdots \ge \lambda_n$ denote the eigenvalues of $M$. 
It has been pointed out that, when the spectral gap $\lambda_1 - \lambda_2$ is 
large enough, then these two centrality measures 
provide very similar rankings, especially when the attention is restricted
to the top edges; on the other hand, different rankings may be
obtained when the gap is small \cite{AB14,BK15}.


\subsection{Edge centralities: directed case}\label{sec:dir}
We now want to define two new edge centrality measures that take into account the directionality of links 
and that can be computed by directly working on the unsymmetric adjacency matrix $A$.

In \cite{AB14} it has been pointed out that one of the main factors in 
the evolution of the total 
communicability is the dominant eigenvalue $\lambda_1$ 
of the matrix involved in its computation.  
This is clear since for an undirected graph with adjacency matrix $A$ the total
communicability can be expressed as
$$TC(A) = \sum_{i=1}^n e^{\lambda_i}\alpha_i^2, \quad \alpha_i = {\bf 1}^T{\bf x}_i,$$
hence the dominant contribution to $TC(A)$ comes from the first term of the sum.
Thus, heuristics that increase the spectral radius of $A$ as much as possible
will likely be effective also when the goal is to increase the total 
communicability as much as possible.
For example, one of the methods found in \cite{AB14} to 
have the best performance relies on 
the edge eigenvector centrality, which is indeed directly connected 
to the change that occurs in the magnitude of the leading eigenvalue 
(see \cite{AB14} for more details). 

Transferring this idea to $T_hC(A)$ and $T_aC(A)$, it follows that we want to 
define (if possible) an edge centrality measure that allows us to control 
the change in the leading singular value of $A$, which corresponds to the square root of 
the leading eigenvalue of $AA^T$ and $A^TA$.

\begin{proposition}\label{prop:eig_change}
Let $A$ be the adjacency matrix of a graph. 
Let ${\bf u}_1$ and ${\bf v}_1$ be the hub and authority vectors, respectively. 
Let $\sigma_1$ be the leading singular value of $A$. 
Consider the adjacency matrix of the graph obtained after the addition of the virtual edge $(i,j)$: 
$\tilde{A}=A+{\bf e}_i{\bf e}_j^T$.
Then the leading eigenvalue $\tilde{\sigma}_1^2$ of the new hub and authority matrices  
satisfies
\begin{equation}\label{eq:eig_up}
\tilde{\sigma}_1^2\geq\sigma_1^2+2\sigma_1 u_1(i)v_1(j)+\max\left\{u_1(i)^2,v_1(j)^2\right\}.
\end{equation}
The inequality is strict if $AA^T$ is irreducible.

Moreover, let $\widehat{A}=A-{\bf e}_i{\bf e}_j^T$ denote 
the adjacency matrix obtained after the 
removal of the existing edge $i\rightarrow j$. 
Then the leading eigenvalue $\widehat{\sigma}_1^2$ of the new hub and authority matrices 
satisfies
\begin{equation}\label{eq:eig_down}
\sigma_1^2 \geq
\widehat{\sigma}_1^2\geq\sigma_1^2-2\sigma_1 u_1(i)v_1(j)+\max\left\{u_1(i)^2,v_1(j)^2\right\}.
\end{equation}
The first inequality is strict if $\widehat{A}\widehat{A}^T$ is irreducible.
\end{proposition}

\begin{proof}

Using the Rayleigh--Ritz Theorem (see, for example, \cite{HJ}) we get:
\begin{align*}
\tilde{\sigma}_1^2=\lambda_1(\tilde{A}\tilde{A}^T) &= \max\limits_{\|{\bf z}\|_2=1}{\bf z}^T\left(\tilde{A}\tilde{A}^T\right){\bf z}\\
				 		   &\geq {\bf u}_1^T\left(\tilde{A}\tilde{A}^T\right){\bf u}_1 \\
						   &= \left\|\left(A^T+{\bf e}_j{\bf e}_i^T\right){\bf u}_1\right\|_2^2 \\
						   &= \|\sigma_1{\bf v}_1 + u_1(i){\bf e}_j\|_2^2 \\
				                   &= \sigma_1^2 + 2\sigma_1u_1(i)v_1(j) + u_1(i)^2. 
\end{align*}
Similarly, by working on the authority matrix one gets:
$$
\tilde{\sigma}_1^2=\lambda_1(\tilde{A}^T\tilde{A})  \geq {\bf v}_1^T\left(\tilde{A}^T\tilde{A}\right){\bf v}_1 
						  			                    = \sigma_1^2 + 2\sigma_1u_1(i)v_1(j) + v_1(j)^2.
$$
From these inequalities, and from 
basic facts from Perron--Frobenius theory,
the conclusion easily follows. Similar arguments
can be used to prove \eqref{eq:eig_down}. 
\qquad 
\end{proof}

Relations~\eqref{eq:eig_up} and \eqref{eq:eig_down} motivate the following definition.
\begin{definition}\label{def:eHITS}
Let $A$ be the adjacency matrix of a directed graph. Let ${\bf u}_1$ and ${\bf v}_1$ be 
its HITS hub and authority vectors, respectively. 
Then the {\rm edge HITS centrality} of the existing/virtual edge $(i,j)$ is defined as
$$^eHC(i,j)=u_1(i)v_1(j).$$
\end{definition}

Notice that when $A$ is symmetric this definition reduces to that of 
edge eigenvector centrality: $^eEC(i,j)=x_1(i)x_1(j)$, where ${\bf x}_1$ is the 
eigenvector associated with the leading eigenvalue of $A$.

\begin{rem}\label{rem:orientation}
{\rm Inequalities~\eqref{eq:eig_up} and~\eqref{eq:eig_down} and, consequently, definition \ref{def:eHITS} 
suggest that there is a ``prescribed direction" one has to follow when introducing a new edge centrality measure. 
Indeed, it is required to use the centrality as broadcaster for the 
source node $i$ and the centrality as receiver for the target node $j$ when evaluating the importance of 
the (virtual) edge $i\rightarrow j$.
This observation confirms a natural intuition and motivates the usage of this same ``orientation'' in all 
our definitions and methods (cf.~section \ref{sec:Heuristics}).}
\end{rem}

The next edge centrality measure we want to define relies on the use of the total communicability of nodes. 
Recall that in the case of an undirected network represented by the symmetric adjacency matrix $M$, the 
total communicability of node $i$ is defined as $(e^M{\bf 1})_i$. 
This quantity describes how well node $i$ communicates with the whole network. 
As discussed in section \ref{sec:TC}, 
this centrality measure is well defined for any adjacency matrix, 
in particular for the adjacency matrices of digraphs, and 
indeed the row and column sums of $e^A$ 
do provide in some cases meaningful measures of how well nodes broadcast information 
(row sums of $e^A$) and 
how good they are at receiving information (column sums of $e^A$). 
However, the expressions describing these quantities do not provide 
information on the alternating walks 
taking place in the digraph and, thus, miss a crucial feature of
communication in real-world directed networks.

For this reason, we introduce here new definitions for the total 
communicabilities of nodes 
which can be shown to be directly connected to their twofold nature.
In order to do so, we make use of the concept of
{\em generalized matrix function} first introduced in \cite{HBI73}.
Let $A = U_r\Sigma_r V_r^T\in {\bb R}^{n\times n}$ be a matrix of rank $r$,
and let $f:{\bb R}\longrightarrow {\bb R}$ be a 
function such that $f(\sigma_i)$ exists for all 
$i=1,2,\ldots,r$, so that the matrix function $f(\Sigma_r)
={\rm diag}(f(\sigma_1), f(\sigma_2), \ldots ,f(\sigma_r))$ is well defined. 

Following \cite{HBI73}, we define
the generalized matrix function 
$f^\diamond:{\bb R}^{n\times n}\longrightarrow {\bb R}^{n\times n}$ as  
\begin{equation*}
f^\diamond(A)=U_rf(\Sigma_r)V_r^T=\sum_{k=1}^rf(\sigma_k){\bf u}_k{\bf v}_k^T.
\end{equation*}

It is easy to check that 
\begin{equation}\label{gen_fun}
f^\diamond(A)=\left(\sum_{k=1}^r\frac{f(\sigma_k)}{\sigma_k}{\bf u}_k{\bf u}_k^T\right)A
=A\left(\sum_{k=1}^r\frac{f(\sigma_k)}{\sigma_k}{\bf v}_k{\bf v}_k^T\right).
\end{equation}

These equalities show that a generalized matrix function can be expressed in terms 
of $A$ and either $AA^T$ or $A^TA$.
Therefore, the entries of $f^\diamond(A)$ --- and hence its row/column sums --- 
can be used as meaningful measures 
of importance in the directed case, {\em provided that they are all non-negative}.

It turns out that, in general, this is not the case for the 
generalized matrix exponential. 
Indeed, consider for example the generalized matrix exponential of the  adjacency matrix
\begin{equation}\label{ex}
A = \left(
\begin{array}{cccc}
 0 & 0 & 1 & 0 \\
 1 & 0 & 0 & 1 \\
 0 & 1 & 0 & 0 \\
 0 & 1 & 0 & 0
\end{array}
\right)
\end{equation}
associated with the digraph in Fig.~1. It turns out that 
its $(3,1)$ and $(4,4)$ 
entries are negative, and thus these quantities cannot be interpreted as 
communicability/centrality measures. 

\begin{figure}[t!]
        \centering
        \label{fig:ex}
        \includegraphics[width=.25\textwidth]{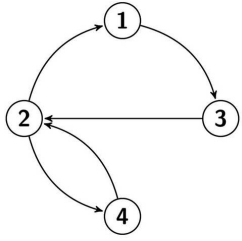}
\caption{The digraph associated with adjacency matrix described in~\eqref{ex}.}
\end{figure}

If we instead consider the generalized hyperbolic sine:
$$\gs(A)=U_r\sinh(\Sigma_r)V_r^T=\sum_{k=1}^r\sinh(\sigma_k){\bf u}_k{\bf v}_k^T,$$ 
we have that this matrix corresponds to the top right block of the 
matrix $e^{\mathscr{A}}$; indeed, we can rewrite equation (\ref{eq:expmA_bipartite}) as
\begin{equation}\label{eq:new_expmA_bipartite}
e^{\mathscr{A}}=\left(
\begin{array}{cc}
\cosh(\sqrt{AA^T}) & \gs(A)\\
\gs(A)^T & \cosh(\sqrt{A^TA})
\end{array}
\right).
\end{equation}
Hence, the entries of $\gs(A)$ are all non-negative, 
and can be used to quantify how
 well nodes communicate when they are playing 
different roles. More precisely, reasoning in terms of alternating walks 
shows that the $(i,j)$th entry of this matrix describes how well node 
$i$ exchanges information 
with node $j$ when the first is playing the role of hub and the latter 
that of authority.
Using this generalized matrix function we can introduce two new centrality 
measures for nodes in digraphs. 

\begin{definition}\label{def:nodeTC}
Let $A=U_r\Sigma_r V_r^T$ be the adjacency matrix of a directed network. 
We call {\rm total hub communicability} of node $i$ the quantity
\begin{subequations}
\begin{equation*}
C_h(i)={\bf e}_i^T\, \gs(A)\, {\bf 1}=
\sum_{k=1}^r\sinh(\sigma_k)({\bf v}_k^T{\bf 1})u_k(i)
\end{equation*}
and {\rm total authority communicability} of node $j$ the quantity
\begin{equation*}
C_a(j)={\bf 1}^T\, \gs(A)\, {\bf e}_j=
\sum_{k=1}^r\sinh(\sigma_k)({\bf u}_k^T{\bf 1})v_k(j)
\end{equation*}
\end{subequations}
\end{definition}

These quantities correspond to row or column sums of the off-diagonal block of 
$e^{\mathscr{A}}$; therefore, $C_h(i)$ 
quantifies the ability of node $i$ --- playing the role of hub --- to communicate with all the nodes in 
the network, when they are all acting as receivers of information. 
Similarly, $C_a(j)$ accounts for the ability of node $j$ as an authority to receive information from all the 
nodes in the graph, when they are acting as broadcasters of
information.\footnote{
The reader is referred  once again to \cite{BEK13} for a more detailed discussion of the
interpretation of the entries in the off-diagonal blocks of $e^{\mathscr{A}}$.
}
This feature highlights the fact that these definitions are better suited than  
$e^A{\bf 1}$ and $({\bf 1}^Te^A)^T$ when it comes to working on digraphs. 
This result is summarized in the following proposition.

\begin{proposition}\label{prop:HAnodeTC}
Let $A$ be the adjacency matrix of a graph $\pazocal{G} = (\pazocal{V},\pazocal{E})$. 
The total hub communicability of node $i\in \pazocal{V}$ can be written as
\begin{subequations}
\begin{equation}\label{eq:HnodeTC1}
C_h(i)=\sum_{k=1}^r\frac{\sinh(\sigma_k)}{\sigma_k}{\bf e}_i^T({\bf u}_k{\bf u}_k^T){\bf d}_{out} = 
\sum_{k=1}^r\frac{\sinh(\sigma_k)}{\sigma_k}({\bf v}_k^T{\bf 1})\sum_{\substack{\ell\in \pazocal{V}\\i\rightarrow \ell}}v_k(\ell).
\end{equation}
Similarly, the total authority communicability of node $j\in \pazocal{V}$ can be expressed as
\begin{equation}\label{eq:AnodeTC1}
C_a(j)=\sum_{k=1}^r\frac{\sinh(\sigma_k)}{\sigma_k}{\bf d}_{in}^T({\bf v}_k{\bf v}_k^T){\bf e}_j = 
\sum_{k=1}^r\frac{\sinh(\sigma_k)}{\sigma_k}({\bf 1}^T{\bf u}_k)\sum_{\substack{\ell\in \pazocal{V}\\ \ell\rightarrow j}}u_k(\ell).
\end{equation}
\end{subequations}
\end{proposition}

\begin{proof}
Using the first equality in~\eqref{gen_fun} one gets that: 
\begin{align*}
C_h(i) 	 & = {\bf e}_i^T\left( \sum_{k=1}^r\frac{\sinh(\sigma_k)}{\sigma_k} {\bf u}_k{\bf u}_k^T\right) A{\bf 1}
	  =  \sum_{k=1}^r\frac{\sinh(\sigma_k)}{\sigma_k} {\bf e}_i^T \left({\bf u}_k{\bf u}_k^T\right) {\bf d}_{out}, 
\end{align*}
which proves the first equality of~\eqref{eq:HnodeTC1}. 
To prove the second, we apply the second equality in \eqref{gen_fun}:
\begin{align*}
C_h(i) 	 & =  ({\bf e}_i^T A)\left( \sum_{k=1}^r\frac{\sinh(\sigma_k)}{\sigma_k}{\bf v}_k{\bf v}_k^T\right){\bf 1}
	  =  \sum_{k=1}^r\frac{\sinh(\sigma_k)}{\sigma_k}({\bf v}_k^T{\bf 1}) A_{i,:}{\bf v}_k,  
\end{align*}
where $A_{i,:}$ is the $i$th row of the adjacency matrix $A$. 
The conclusion then follows from the fact that $A_{i,:}{\bf v}_k = \sum_{ i\rightarrow \ell}v_k(\ell)$.
The proof of~\eqref{eq:AnodeTC1} goes along the same lines and is thus omitted.\qquad
\end{proof}

Before proceeding with the introduction of the associated edge centrality measure, we want to show with a small example that these measures of hub and authority centrality are 
indeed informative. 
Consider as an example the graph in Fig.~1. 
It is intuitive that node $2$ should be given the highest score both 
as hub and as authority by any reasonable centrality measure. 
Consequently, the authority scores for nodes $1$ and $4$ should be 
the same and higher than that of node $3$ because these nodes  
are directly pointed to from node $2$, which is the best hub in the graph. 
For a similar reason, nodes $3$ and $4$ should be ranked higher than 
node $1$ when considering a hub score, since they directly point 
to node $2$, which is the most important authority.

\begin{table}[t]
\begin{center}
\footnotesize
\caption{Centrality measures for the nodes in the graph represented in Fig.~1 
and described by the adjacency matrix~\eqref{ex}.}
\label{tab:ex}
\begin{tabular}{ccccccccc}
\hline
NODE & $d_{out}(i)$ & $d_{in}(i)$ & $u_1(i)^2$ & $v_1(i)^2$  & $C_h(i)$ & $C_a(i)$\\
\hline
1    & 1            &  1	  & .0000    & .3333     & 1.1752 & 1.3683 \\
2    & 2            &  2	  & .5000    & .3333     & 2.7366 & 2.7366 \\
3    & 1            &  1	  & .2500    & .0000     & 1.3683 & 1.1752 \\
4    & 1            &  1	  & .2500    & .3333     & 1.3683 & 1.3683 \\
\hline
\end{tabular}
\end{center}
\end{table}
  
Table~\ref{tab:ex} contains the centrality scores for the 
four nodes when the in/out-degree, HITS centrality\footnote{To compute 
these scores, we initialize the HITS algorithm with the constant authority 
vector with 2-norm equal to 1; see \cite{K99,BEK13}.}, 
and the total hub/authority communicability are considered. 
Clearly, the in/out-degrees of the nodes do not capture the picture we just described  
since they cannot discriminate between nodes 1, 3, and 4.
This happens because the degree centralities take into account only local information 
about how nodes propagate 
information in the network. 

Concerning HITS,
the rankings given by the hub scores conform to our expectations, but those
given by the authority scores do not, since they are unable to identify node 2
as the most authoritative one (it is tied with nodes 1 and 4). 
Another problem with HITS is that the rankings will
depend in general on the initial vector, since for this example the matrices
$AA^T$ and $A^TA$ are reducible (this also explains the occurrence of zero entries in
the hub and authority vectors). Note that this is a non-issue for both $C_h(i)$
and $C_a(i)$; most importantly, however, these two measures succeed in identifying
the ``correct" relative rankings for the hubs and authorities in this digraph.

These observations motivate the introduction of 
a new edge centrality measure.
\begin{definition}\label{def:eTC}
Let $A$ be the adjacency matrix of a simple digraph. 
Then the {\rm edge total communicability centrality} of the existing/virtual edge $(i,j)$ is defined as
$$^egTC(i,j)=C_h(i)C_a(j),$$
where $C_h(i)$ and $C_a(j)$ are the total hub communicability of node $i$ and the total authority communicability of node $j$, respectively. 
\end{definition} 

Note that when the difference between the two largest singular values $\sigma_1-\sigma_2$ is ``large enough'', the 
quantities $C_h(i)$ and $C_a(j)$ are essentially determined  by $\sinh(\sigma_1)\|{\bf v}_1\|_1u_1(i)$ 
and $\sinh(\sigma_1)\|{\bf u}_1\|_1v_1(j)$, respectively. 
When this condition is satisfied we expect agreement between the rankings for the edges provided by the edge HITS and total 
communicability centrality measures, at least when the attention is restricted to the top ranked edges.

It is natural to ask how the edge centrality measure just introduced is related to 
the edge total communicability centrality applied to the undirected graph $\mathscr{G}$.
For the centrality of the (virtual) edge 
$(i,j)$ we obtain 
\begin{equation}\label{eq:diff}
{^e{\pazocal TC}(i,j')} - [{^egTC(i,j)}] =  \phi(i,j) - \left (\cosh(\sqrt{AA^T}){\bf 1}\right )_i
\left (\cosh(\sqrt{A^TA}){\bf 1}\right )_j
\end{equation}
where ${^e{\pazocal TC}(i,j')}$ is the edge total communicability of $(i,j')$ in the bipartite graph $\mathscr{G}$, 
$j'=j+n$, and 
$$\phi(i,j) = (e^{\mathscr{A}}{\bf 1})_i\left (\cosh(\sqrt{A^TA}){\bf 1}
\right )_j+(e^{\mathscr{A}}{\bf 1})_{j'}\left (\cosh(\sqrt{AA^T})
{\bf 1}\right )_i.$$
The difference in \eqref{eq:diff} 
is positive and it may be so large that the edge 
selected when working on the digraph could well
be different from that selected when working on the associated bipartite network, 
thus leading to different results  for the two techniques. 
As we will
see in the section on numerical experiments, the two criteria may indeed lead to
different results.

\begin{rem}
{\rm Concerning the actual computation of the quantities that 
occur in Definition~\ref{def:nodeTC},  one can either exploit the
relationship (\ref{eq:new_expmA_bipartite}) between $e^{\mathscr{A}}$ and $\gs(A)$
and use
standard methods for computing the matrix exponential \cite{fun_matrix} or,
if the matrix $A$ is too large to build and work with $\mathscr{A}$ explicitly, 
one can obtain estimates of the quantities of interest using 
the Golub--Kahan algorithm \cite{GM,GVL}. 
Indeed, $\gs(A)$ can be rewritten as  
\begin{equation*}\label{gen_fun2}
\gs(A)=\sinh(\sqrt{AA^T})(\sqrt{AA^T})^\dagger A=A(\sqrt{A^TA})^\dagger \sinh(\sqrt{A^TA}),
\end{equation*}
where ``$\dag$" denotes the Moore--Penrose pseudoinverse,
and one can obtain estimates of the desired
row and column sums by applying 
Golub--Kahan bidiagonalization with an appropriate starting vector 
($A{\bf 1}$ or $A^T{\bf 1}$. respectively).
We plan to investigate these and other computational issues in future work. 
The test matrices used in this paper are small enough that we could 
form and manipulate the matrix $\mathscr{A}$ explicitly. 
Therefore, we expect the heuristics based on the two edge centrality 
measures ${^egTC}(i,j)$ and ${^e\pazocal{TC}}(i,j')$ to perform 
similarly in terms of timings.
}
\end{rem}


\section{Heuristics}\label{sec:Heuristics}

In this section we describe the methods we will use to perform the 
numerical tests presented in section~\ref{sec:test}.
For both the updating and downdating problem, we will first rank the (virtual) edges
 using a variety of edge centrality measures; for large
graphs we may consider only a subset of all possible candidate edges, as
discussed below.
For the updating problem, we will then select the top ranked virtual edges, 
while for the the downdating problem we will select the edges having the lowest 
centrality rankings. 
Given a budget of $K$ modifications to be performed, we can proceed in one of two ways.
We can either perform one edge modification at a time and then recalculate all the
necessary centrality scores right afterwards, or we can perform all the 
modifications at once, without recalculation. This latter 
approach will correspond to the {\tt .no} variants 
of the algorithms. In the undirected case, the latter approach was found to be
essentially as effective as the former (even for relatively large $K$) while
being dramatically less expensive in terms of computational effort; see \cite{AB14}.

As we already mention in section \ref{sec:vs},
we can either work on the bipartite network associated with the digraph 
or directly on the original network. 
When working on the original graph, addition/deletion of an edge
corresponds to rank-one updates/downdates 
to the corresponding adjacency matrix $A$. 

The methods used are labeled as follows:
\begin{itemize}
\item {\tt eig(.no)}. Let ${\bf x}_1$ be the right eigenvector associated with the 
leading eigenvalue of $A$ (assumed to be simple) 
and ${\bf y}_1$ be the left eigenvector associated with the same eigenvalue. 
Generalizing the definition for the edge eigenvector centrality given in section \ref{sec:und}, we can define in 
the case of digraphs: $$^eEC(i,j):=x_1(i)y_1(j).$$

This quantity has been recently used in \cite{TPEFF} to devise algorithms 
aimed at increasing as much as possible the 
leading eigenvalue of $A$ when edges are added to the network.

\item {\tt TC(.no)}. Here we use the total communicability $e^A{\bf 1}$. 
The score assigned to a (virtual) edge $(i,j)$ is: $$^eTC(i,j):=\left(e^A{\bf 1}\right)_i\left({\bf 1}^Te^A\right)_j.$$
This heuristic generalizes to digraphs the analogous one for
undirected graphs (cf.~{\tt nodeTC(.no)} in \cite{AB14}).
\item {\tt HITS(.no)}. Each (virtual) edge is given a score in terms of 
the quantities introduced in Definition~\ref{def:eHITS}:
$$^eHC(i,j)=u_1(i)v_1(j)$$
\item {\tt gTC(.no)}. This heuristic is based on the edge total communicability defined in terms of the 
generalized hyperbolic sine (see definition~\ref{def:eTC}). 
The (virtual) edge $(i,j)$ is assigned the score: $$^egTC(i,j) = C_h(i)C_a(j),$$
where $C_h(i)=(\gs(A){\bf 1})_i$ and $C_a(j)=(\gs(A^T){\bf})_j$.
\end{itemize}

The first two methods (with their variants) generalize to the case of digraphs the techniques which performed best  
in the undirected case. 
Notice that we have used the broadcaster score for the source node and the receiver score for the target node 
(see Remark~\ref{rem:orientation}). 

Next, we
consider the bipartite network associated to the matrix $\mathscr{A}$ 
defined in \eqref{eq:A_bipartite}. 
The criteria we use to select the modifications are 
based on the edge centrality measures described in section \ref{sec:und}. 
We will label the methods as follows:
\begin{itemize}
\item {\tt b:eig(.no)}. We use the eigenvector centrality of edges; 
the edge eigenvector centrality of the (virtual) edge $(i,j')$ 
is defined as 
$$^e{\pazocal EC}(i,j')=q_1(i)q_1(j'),$$
where ${\bf q}_1$ is the Perron vector of $\mathscr{A}$.
\item {\tt b:TC(.no)}. This is based on the total communicability centrality of edges: each (virtual) edge $(i,j')$ is assigned the score:
$${^e{\pazocal TC}(i,j')}=\left(e^\mathscr{A}{\bf 1}\right)_i\left(e^\mathscr{A}{\bf 1}\right)_{j'}.$$

\item {\tt b:deg}. This simple heuristic is
 equivalent to the {\tt degree} method in \cite{AB14}. Each (virtual) edge is assigned a score of the form: 
$$d(i)+d(j'), \quad i\in\pazocal{V} \text{ and } j'\in\pazocal{V}',$$
where $d(i) = (\mathscr{A}{\bf 1})_i$ is the degree of node $i$ in 
the network represented by $\mathscr{A}$.
\end{itemize}

\begin{rem}\label{rem:deg}
{\rm 
We do not provide a method that generalizes {\tt degree} in~\cite{AB14} to the case of digraphs since it would 
coincide with the heuristic {\tt b:deg} just introduced. 
Indeed, the straightforward generalization would require to assign to the (virtual) edge 
$i\rightarrow j$ the score $d_{out}(i)+d_{in}(j)$. 
However, it is easy to see that 
$d_{out}(i) = d(i)$ where $i\in{\pazocal V}$ and $d_{in}(j) = d(j')$ where $j'\in\pazocal{V}'$, and thus this 
technique would be indistinguishable from {\tt b:deg}. 
Note that this technique is the optimal one if we want to optimize the sum $T_hC(A) + T_aC(A)$ 
and we use the second order Maclaurin approximations 
$\cosh(\sqrt{X}) \approx I + \frac{X}{2}$, with $X = AA^T,A^TA$ to 
compute the the total hub and authority communicabilities.  
}
\end{rem}

When working on the matrix associated with the bipartite graph, each edge 
modification of the corresponding network will 
cause a rank-two change in $\mathscr{A}$.
We want to stress once again that the set of (virtual) edges among 
which to select the modifications is the same whether we 
work on $A$ or on $\mathscr{A}$ and corresponds to the set of 
(virtual) edges of the graph $\pazocal G$, or a subset of it. 
For large networks, the set of virtual edges among 
which to select the updates may be too large to be exhaustively searched. 
In this work we used the whole set for all the networks used in the
experiments except the largest one, namely cit-HepTh (see table \ref{tab:dataset}). 
For this problem, we restrict the search to a subset of the set of
all virtual edges constructed as follows.
We first rank in descending order the nodes of $\mathscr{G}$ using 
the eigenvector centrality. 
This results in a ranking of $2n$ elements: the nodes in $\pazocal{V}$ and their copies. 
Next, for each $i=1,\ldots,n$ we remove from the list the one element between $i$ and its 
copy $i'$ which has the lowest rank. 
We now have a list of length $n$ which includes either one element (element of $\pazocal{V}$) or 
its copy (element of $\pazocal{V}'$). 
We thus relabel all the copies, if present, with the label
of the corresponding node in $\pazocal{V}$. 
The resulting list contains all the $n$ nodes in the original graph. It  
has been obtained considering, for each node, its best performance 
between its role as hub and its 
role as authority in the network.
Finally,
we take the induced subgraph corresponding to the top $10\%$ of the nodes in 
this list.
The set of virtual edges in this subgraph is the set we exhaustively search.

\subsection{Rank-two modifications}\label{ssec:rank2}

Before discussing the results obtained by applying our techniques to select rank-one updates of the matrix $A$, we 
want to briefly discuss how these techniques may be modified in order to make them suitable to select symmetric rank-two 
modifications of the unsymmetric adjacency matrix. 
This approach goes beyond the scope of this paper, but it is worth some discussion. 
Indeed, in real world applications one may conceivably want to 
add (or delete) two-directional edges between nodes in a digraph
in order to tune its total communicabilities. 
In this setting, the downdating and updating problems aim at the same goals as before, 
but the sets in which one searches 
for modifications are different from those used in our original problems. 
Indeed, the updates will be selected in the set $\{(i,j)\in  \pazocal{V}\times \pazocal{V} | (i,j),(j,i)\not\in\pazocal{E}\}$, 
while the downdates will be selected among the edges in $\{(i,j)\in  \pazocal{V}\times \pazocal{V} | (i,j),(j,i)\in\pazocal{E}\}$. 

We start by discussing the case of the degree and of the edge HITS centrality. 
The results obtained for these two approaches will motivate the generalization of the other techniques. 
As we have observed in Remark~\ref{rem:deg}, the degree strategy works as the optimal strategy when we consider a second order approximation 
of the terms in the sum $T_hC(A)+T_aC(A)$. 
By carrying out the same computation, replacing a rank-one update of the adjacency matrix with a rank-two update, one finds that the most 
natural generalization requires that the quantities used to  rank the (virtual) edges by the method based on the degree of nodes are 
$$[d_{in}(i) + d_{out}(j)] + [d_{out}(i) +d_{in}(j)].$$
A similar results can be obtained if we want to adapt {\tt HITS} to handle rank-two updates. Indeed, to rank the undirected (virtual) edges one 
may use the quantities 
$$^eHITS(i,j) + {^eHITS(j,i)} .$$
This follows from the application  to 
the matrices $(A + {\bf e}_i{\bf e}_j^T + {\bf e}_j{\bf e}_i^T)(A + {\bf e}_i{\bf e}_j^T + {\bf e}_j{\bf e}_i^T)^T$ and 
$(A + {\bf e}_i{\bf e}_j^T + {\bf e}_j{\bf e}_i^T)^T(A + {\bf e}_i{\bf e}_j^T + {\bf e}_j{\bf e}_i^T)$ 
of the same techniques used in the proof of Proposition~\ref{prop:eig_change}. 

From these simple results, it follows that the quantities used by the other heuristics to handle rank-two modifications of the adjacency matrix 
of a digraph have the form 
$$^eC(i,j) + {^eC(j,i)},$$
where $^eC$ is one among the edges centralities used in the previous section to work in the directed case.


\section{Numerical tests}\label{sec:test}

\begin{table}[t]
\centering
\footnotesize
\caption{Description of the dataset.}
\label{tab:dataset}
\begin{tabular}{ccccccc}
\hline
NETWORK	          & $n$    & $m$     & $\tau$   & $\sigma_1$ & $\sigma_2$ & $\sigma_1-\sigma_2$   \\
\hline 
GD95b             &  73    &   96    &     5160 &  4.79      & 4.37       & 0.428                 \\
Comp.~Complexity  & 857    & 1596    &   731996 & 10.93      & 9.87       & 1.05                  \\
Abortion          & 2262   & 9624    &  5104728 & 31.91      & 20.04      & 5.87                  \\
Twitter           & 3656   & 188712  & 13176871 & 189.15     & 120.54     & 68.71                 \\ 
cit-HepTh         & 27400  & 352547  &  3730367 & 85.16      & 69.31      & 15.85                 \\   
\hline
\end{tabular}
\end{table}

The numerical tests have been performed on five networks, which come from three sources. 
The small network GD95b comes from the University of Florida sparse matrix
collection \cite{DataF} and represents entries in a graph drawing context.
The citation network cit-HepTh, the largest one in our data set,
also comes from the University of Florida sparse matrix collection
\cite{DataF}.
The networks Abortion and Computational Complexity are small web graphs consisting of web sites on the topic of abortion and  computational complexity. They are available online 
at \cite{DataT}. 
Finally, the network Twitter can be found at \cite{DataG}; it contains mentions and retweets of some part of the social network Twitter. 
Table \ref{tab:dataset} summarizes some properties of the networks in our dataset; namely, it contains the number of nodes $n$ and edges $m$, 
the two largest singular values of the adjacency matrix $\sigma_1$, $\sigma_2$, their difference $\sigma_1-\sigma_2$, and the number of virtual edges $\tau$. 
An exception is the network cit-HepTh, for which 
$\tau$ is the number of 
virtual edges contained in the subgraph of the network constructed as described 
at the end of 
section~\ref{sec:Heuristics}. 

\begin{figure}[t]
\centering
\includegraphics[width=.83\textwidth]{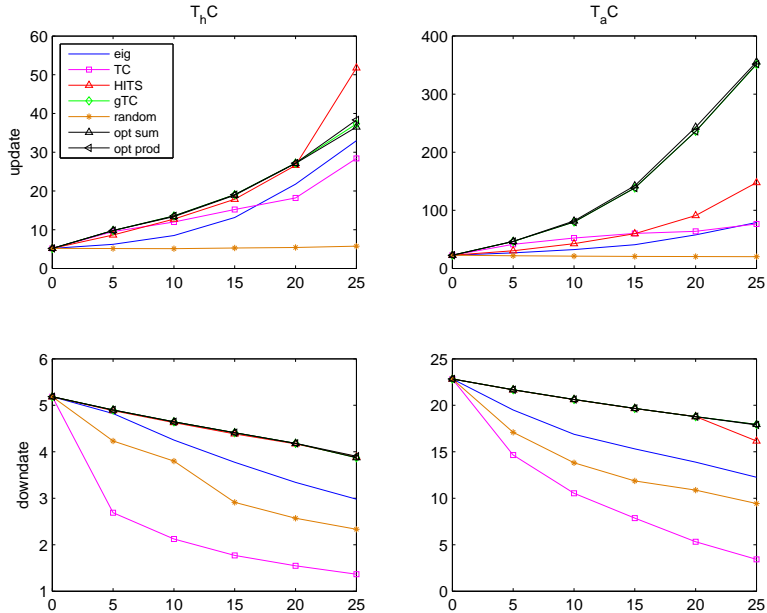}
\caption{Evolution of $T_hC$ and $T_aC$ for the network GD95b when 25 
edge modifications are performed working on the 
matrix $A$ associated with the digraph: updates (top) and
 downdates (bottom).}
\label{fig:small}
\end{figure}

The small network is used to compare the effectiveness of the proposed
heuristics with a ``brute force" approach where each virtual edge is
added in turn and the change in total communicability is monitored in
order to find the ``optimal" choice. Since we are tracking not one but two
quantities, $T_hC(A)$ and $T_aC(A)$, we monitor both
$T_hC(A) + T_aC(A)$ and $T_hC(A)\cdot T_aC(A)$ and choose the optimal
edge for either one of them. 
 These methods are labeled as {\tt opt sum} and {\tt opt prod}, respectively.
We perform a similar set of experiments
for the downdating. As a baseline method, we also report results for a
random   selection of the edges in all our tests. 
The random methods are labeled as {\tt random} or {\tt b:random}, depending on 
whether we work on the matrix $A$ or on $\mathscr{A}$. 

In Fig.~\ref{fig:small} we show plots of the total communicabilities $T_hC$ and $T_aC$  
when up to $K=25$ edge modifications are performed. 
We limit ourselves
to the results for the heuristics based on the original digraph (matrix $A$).
The results show that the heuristic {\tt gTC} performs as well as the ``optimal"
choice based on brute force, while of course being much less expensive, in tackling both the updating and downdating problem. 
Note moreover that the performance of the methods {\tt HITS} and {\tt gTC} 
is different for this network. 
This result agrees with what one would expect, in view of the small gap 
$\sigma_1-\sigma_2$ of the adjacency matrix under study.
When considering the problem of downdating, on the other hand, all 
the methods perform well. In particular we want to stress again 
the excellent performance of the method {\tt gTC}.  
The only exception is perhaps the heuristic {\tt eig}, whose performance for the first 5 steps is comparable with the random choice.  
This result confirms our  
claim that this heuristic, which was shown in~\cite{AB14} to work
very well for undirected networks, is not a good approach in the directed case.

\begin{figure}[t]
\centering
\includegraphics[width=.83\textwidth]{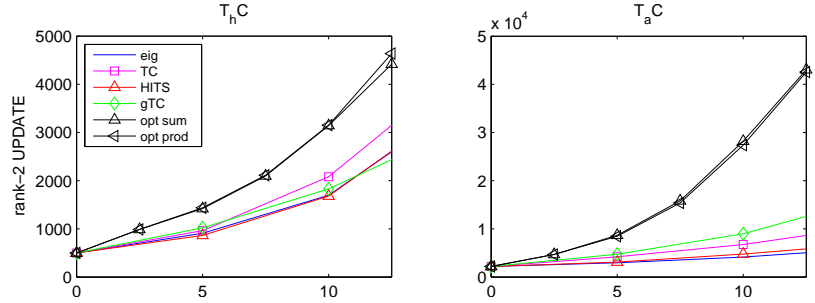}
\caption{Evolution of $T_hC$ and $T_aC$ for the network GD95b when 25 
symmetric edge modifications are performed working on the 
matrix $A$ associated with the digraph. The optimal methods refer to the rank-one selection of the modifications}
\label{fig:small2}
\end{figure}

In Fig.~\ref{fig:small2} we display the evolution of the total communicability indices under rank-two updates. 
In this plot we retain the same names for the techniques as used in case of the rank-one modifications; 
however, the quantities used to derive the rankings 
are defined as in subsection~\ref{ssec:rank2}. 
In this figure, each step corresponds to a rank-two symmetric modification, for 
the heuristic based on the edge centrality measures, and to two rank-one modifications, for the optimal methods. 
Thus, the plots for the optimal methods coincide with those in Fig.~\ref{fig:small}. 
The results displayed in Fig.~\ref{fig:small2} tell us that the symmetric 
rank-two modifications of the matrix may not lead to results as good as those obtained with the rank-one updates. 
Indeed, for both the total hub and authority communicabilities we have at least three methods in Fig.~\ref{fig:small} that 
outperform all the methods used in Fig.~\ref{fig:small2}. 
For this reason, we have not further investigate this approach.


\begin{figure}[t]
\centering
\includegraphics[width=.83\textwidth]{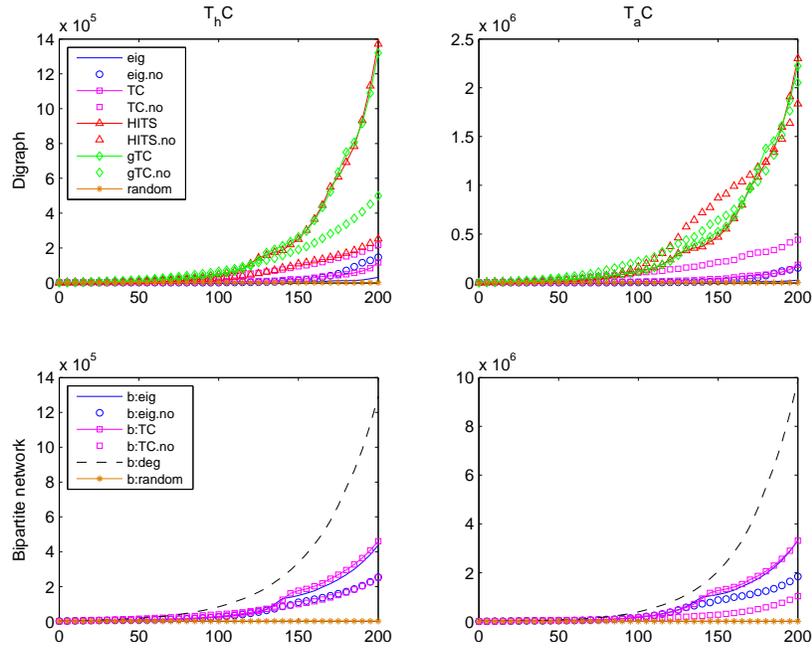}
\caption{Evolution of $T_hC$ and $T_aC$ for the network Computational Complexity when 200 updates are selected working on the 
matrix $A$ associated with the digraph (top) and
 on its bipartite version $\mathscr{A}$ (bottom).}
\label{fig:CC_up}
\end{figure}

\begin{figure}[h!]
\centering
\includegraphics[width=.83\textwidth]{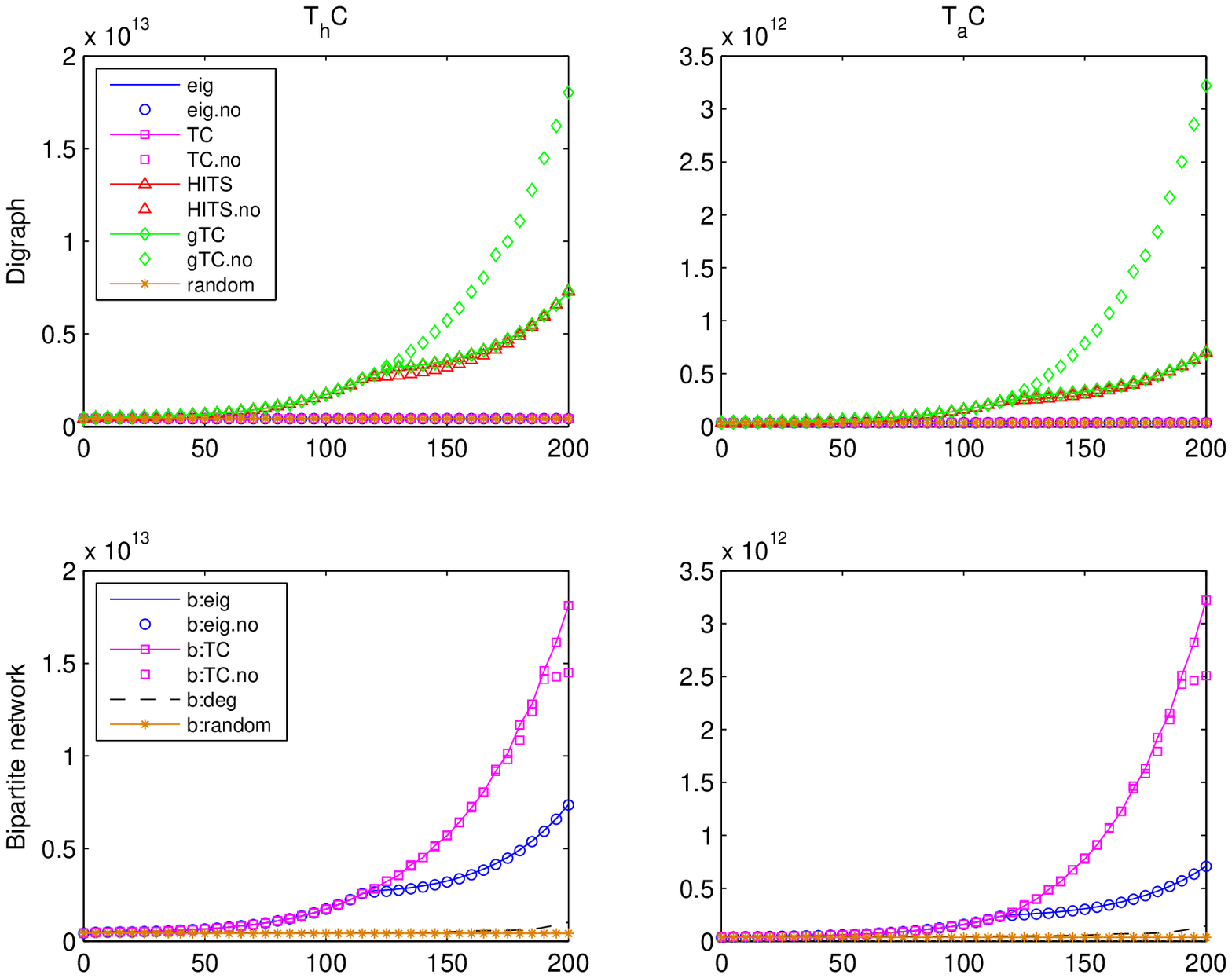}
\caption{Evolution of $T_hC$ and $T_aC$ for the network Abortion when 200 updates are selected working on the 
matrix $A$ associated with the digraph (top) or on its bipartite version $\mathscr{A}$ (bottom).}
\label{fig:Abortion_up}
\end{figure}

\begin{figure}[t]
\centering
\includegraphics[width=.83\textwidth]{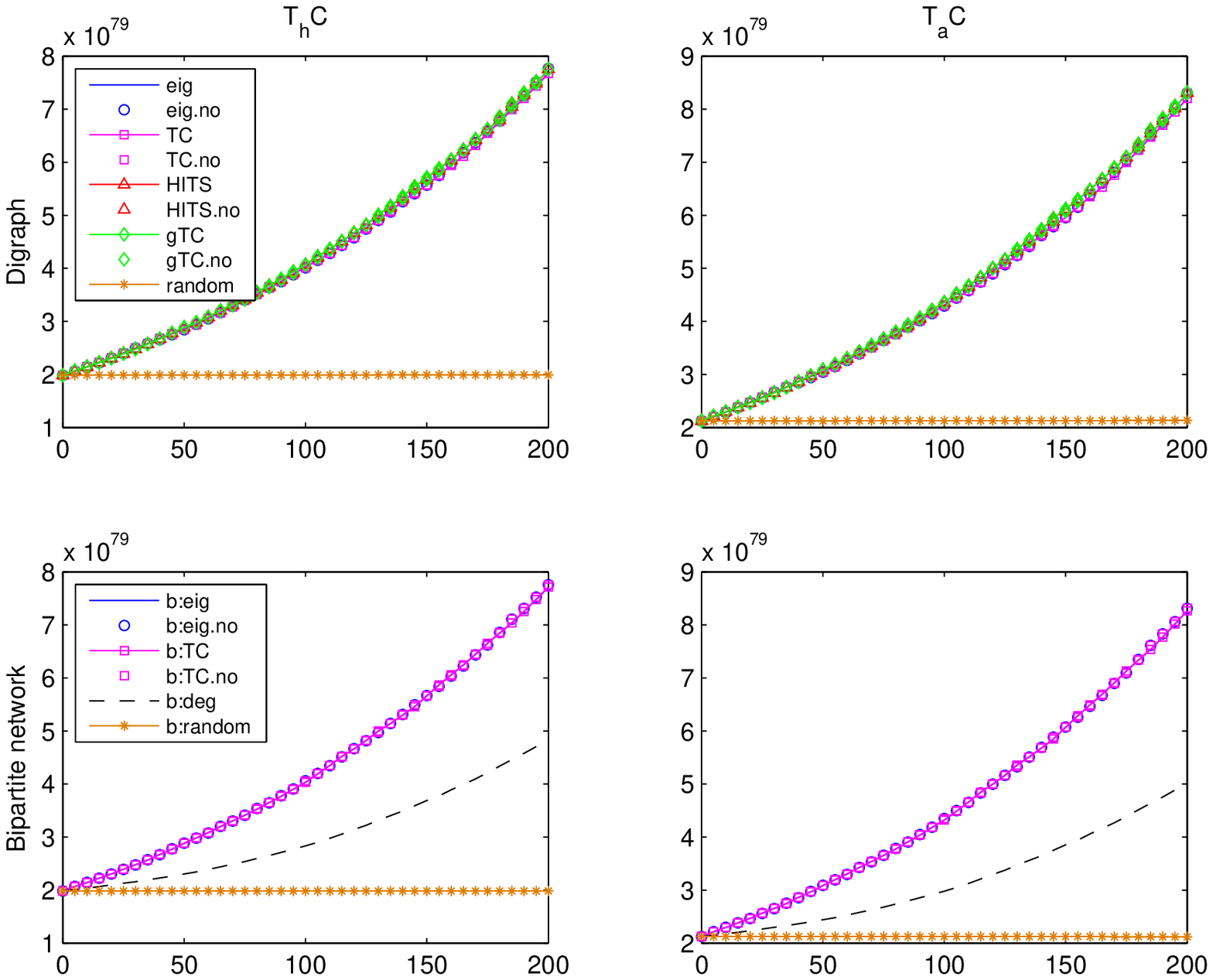}
\caption{Evolution of $T_hC$ and $T_aC$ for the network Twitter when 200 updates are selected working on the 
matrix $A$ associated with the digraph (top) and on its bipartite version $\mathscr{A}$ (bottom).}
\label{fig:Twitter_up}
\end{figure}

\begin{figure}[h!]
\centering
\includegraphics[width=.83\textwidth]{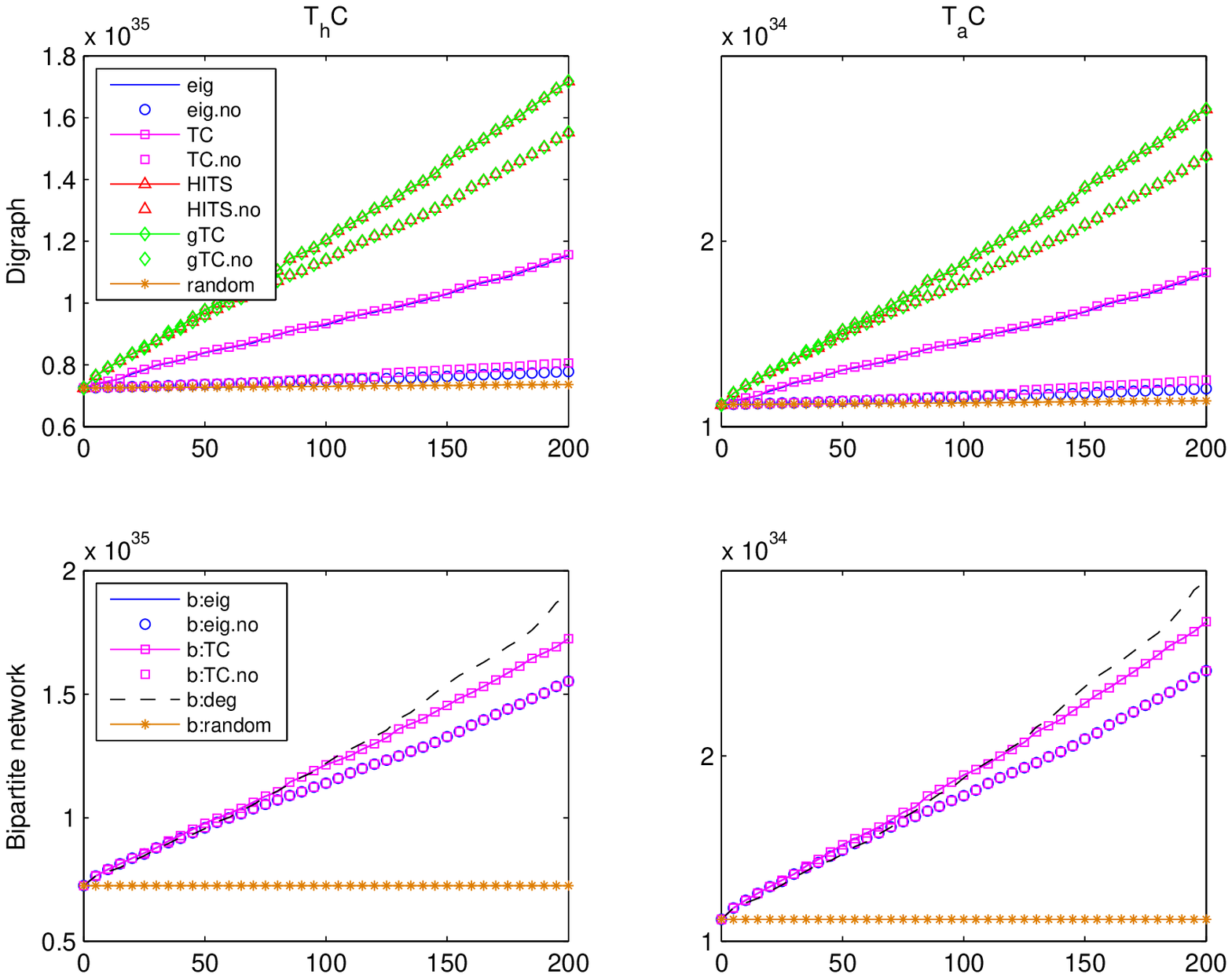}
\caption{Evolution of $T_hC$ and $T_aC$ for the network cit-HepTh when 200 updates are selected working on the 
matrix $A$ associated with the digraph (top) or on its bipartite version $\mathscr{A}$ (bottom).}
\label{fig:cit-HepTh_up}
\end{figure}

The results on the small network give us confidence that at least some of our proposed
heuristic do a very good job at enhancing the communicability properties of digraphs.
In the remaining tests we concentrate on the larger four networks, for which the ``optimal,"
brute force approaches are not practical.
All experiments were performed using MATLAB Version 8.0.0.783 (R2012b)
on an IBM ThinkPad running Ubuntu 14.04 LTS, a 2.5 GHZ Intel Core i5 processor, and
3.7 GiB of RAM.

Figs.~\ref{fig:CC_up}-\ref{fig:cit-HepTh_up} display the evolution of the total hub and communicability centrality (rescaled by the 
number of edges in the network)  
when $K=200$ updates are selected using the criteria previously introduced. 
The plots at the top of each figure display the evolution of the total hub communicability (left) and total authority communicability (right) when the digraph is modified using the techniques developed for the directed case. 
The bottom plots show the evolution of the two indices obtained when 
the modifications are selected by working on $\mathscr{A}$. 
As expected, the proposed heuristics are dramatically better
than the random choice.

The results show that the heuristics {\tt b:eig(.no)} and {\tt b:TC(.no)} perform
similarly to {\tt HITS(.no)} and 
{\tt gTC(.no)}. 
The methods {\tt eig(.no)} and {\tt TC(.no)} display erratic behavior
and often perform very poorly,  
as shown in Figs.~\ref{fig:CC_up},~\ref{fig:Abortion_up}, and~\ref{fig:cit-HepTh_up}. 
The method {\tt eig(.no)} also suffers from the restriction that the 
dominant eigenvalue must
be simple, which is not always true in practice.
Likewise,
the performance of {\tt b:deg} is generally unsatisfactory,
with the exception of $T_aC(A)$ for 
the network Computational
Complexity where
 it outperforms the other techniques 
(see Fig.~\ref{fig:CC_up}). 
Overall, considering also the timings (see Table \ref{tab:time_Updt}),
the best performance is displayed by the heuristics {\tt gTC(.no)} and {\tt HITS(.no)} and by their undirected counterparts {\tt b:TC(.no)} and {\tt b:eig(.no)}. 
The only possible exception is the Computational Complexity network, for which the heuristics 
for the directed case outperform those for the undirected, bipartite counterpart.

The disagreement between the results for the heuristics labeled {\tt HITS} 
and {\tt b:eig} for the network Computational Complexity is at first sight puzzling. The two 
criteria should lead to the same edge selection and therefore to the same
results, since 
the principal eigenvector of $\mathscr{A}$ is ${\bf q}_1=({\bf u}_1^T, {\bf v}_1^T)^T$ 
and thus $q_1(i)=u_1(i)$ and $q_1(j')=v_1(j)$ in the definition of the 
heuristic {\tt b:eig}. 
However, if at least two edges have the same centrality score when 
working with {\tt b:eig} and {\tt HITS}, then the two methods may select 
different edges. In this case, after the edge modification has been performed, 
the adjacency matrices manipulated by the two methods are different, 
thus causing the difference we observe in Fig.~\ref{fig:CC_up}.
The difference will be more pronounced if the tie between edges occur
at the beginning of the modification process.

\begin{table}[t]
\footnotesize
\centering
\caption{Timings in seconds when $K=200$ updates are selected for the networks in our Dataset using the methods described.}
\label{tab:time_Updt}
\begin{tabular}{lcccc}
\hline
                 &     Computational        &          &         &           \\
                 &      Complexity          & Abortion & Twitter & cit-HepTh \\
\hline
{\tt eig}        &       12.51              &   53.27  &  139.82 &   217.12  \\
{\tt eig.no}     &        0.13              &    0.73  &    1.75 &     1.33  \\
{\tt TC}         &      114.67              &   62.22  &  187.22 &   163.55  \\
{\tt TC.no}      &        0.61              &    0.76  &    2.19 &     1.02  \\
{\tt HITS}       &        8.35              &   50.69  &  133.50 &    88.82  \\
{\tt HITS.no}    &        0.09              &    0.63  &    1.69 &     0.67  \\
{\tt gTC}        &       10.63              &   59.31  &  183.48 &   205.17  \\
{\tt gTC.no}     &        0.12              &    0.68  &    1.77 &     1.28  \\
{\tt b:eig}      &        9.35              &   52.43  &  134.03 &    99.70  \\
{\tt b:eig.no}   &        0.21              &    0.69  &    1.66 &     0.88  \\
{\tt b:deg}      &       11.00              &   85.06  &  256.66 &    84.95  \\
{\tt b:TC}       &       11.39              &   59.31  &  154.97 &   139.50  \\
{\tt b:TC.no}    &        0.11              &    0.72  &    1.67 &     0.82  \\
\hline
\end{tabular}
\end{table}

Table~\ref{tab:time_Updt} contains the timings (in seconds) employed for the 
selection of the $K=200$ virtual edges to be updated.
The heuristics used were implemented using mostly built-in MATLAB functions, such
as the function {\tt eigs} used for computing the largest eigenvalue.
For the heuristics requiring the computation of a matrix function times a vector
we used the code {\tt funm\_kryl} by S.~G\"uttel \cite{Guettel}.
To implement the degree-based heuristic we wrote our own code, which is far
from optimal when compared to the other ones. The relatively high timings reported
for this heuristic can likely be reduced with a more careful implementation.
When interpreting the results, it has to be kept in mind that the size $\tau$ of the set of virtual edges can be pretty large (cf. Table~\ref{tab:dataset}). 
We have observed that, for all the methods, roughly half of 
the reported computing time
is spent in the computation of the products used in the definitions of the edge centrality measures. Nevertheless, the timings range from very small
to moderate in all cases, showing the feasibility of the proposed heuristics. 

Among all the methods we tested on directed networks for the updating problem, the best performance is displayed by {\tt HITS(.no)}, {\tt gTC(.no)}, {\tt b:eig(.no)} and {\tt b:TC(.no)}  
with the methods that manipulate $A$ having the edge when $\sigma_1-\sigma_2$ is small. 
Due to its erratic behavior, we cannot recommend the use of {\tt b:deg} in general.


\begin{figure}[t!]
\centering
\includegraphics[width=.83\textwidth]{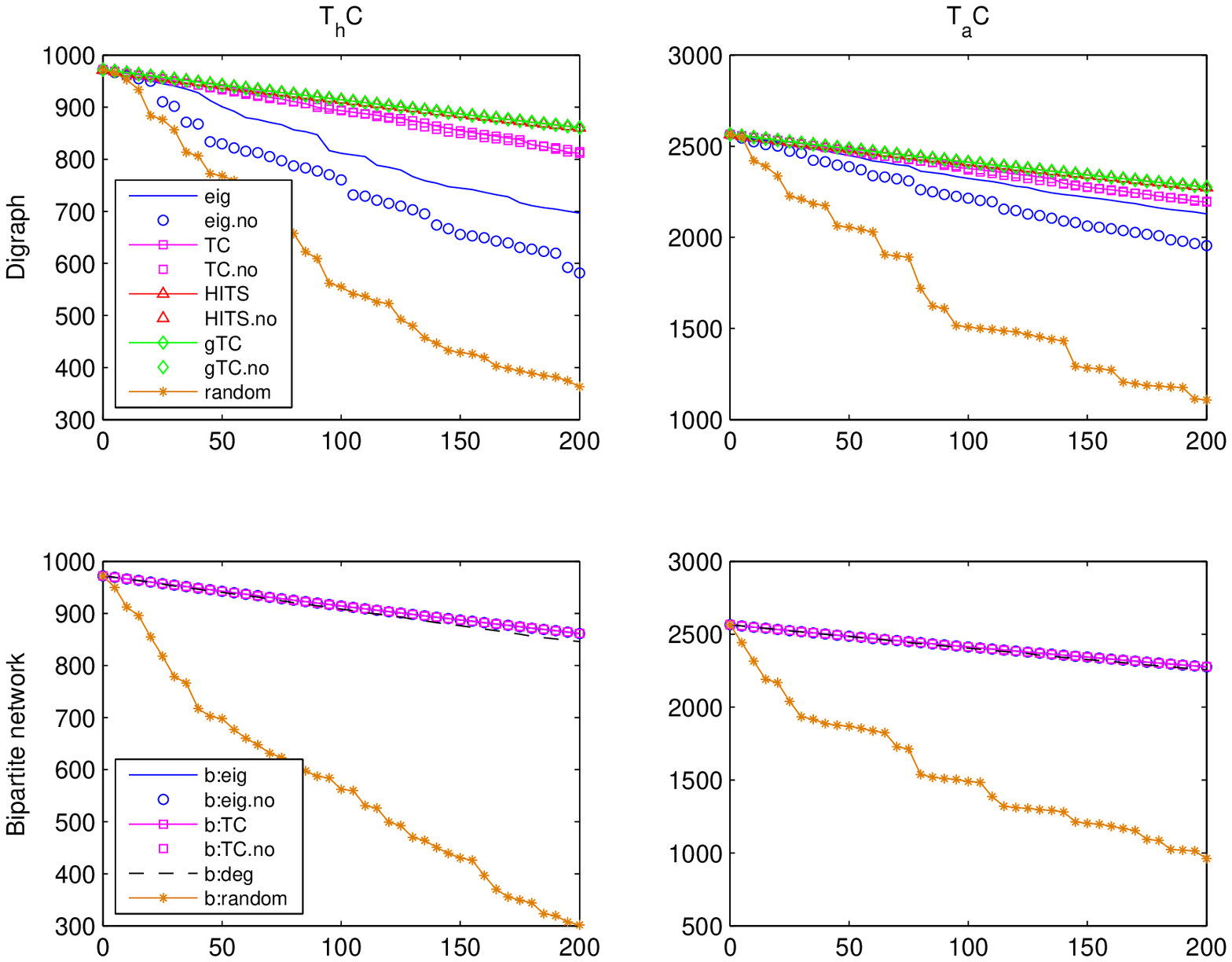}
\caption{Evolution of $T_hC$ and $T_aC$ for the network 
Computational Complexity when 200 downdates are selected working on the 
matrix $A$ associated with the digraph (top) and on its bipartite version $\mathscr{A}$ (bottom).}
\label{fig:CC_down}
\end{figure}

\begin{figure}[h!]
\centering
\includegraphics[width=.83\textwidth]{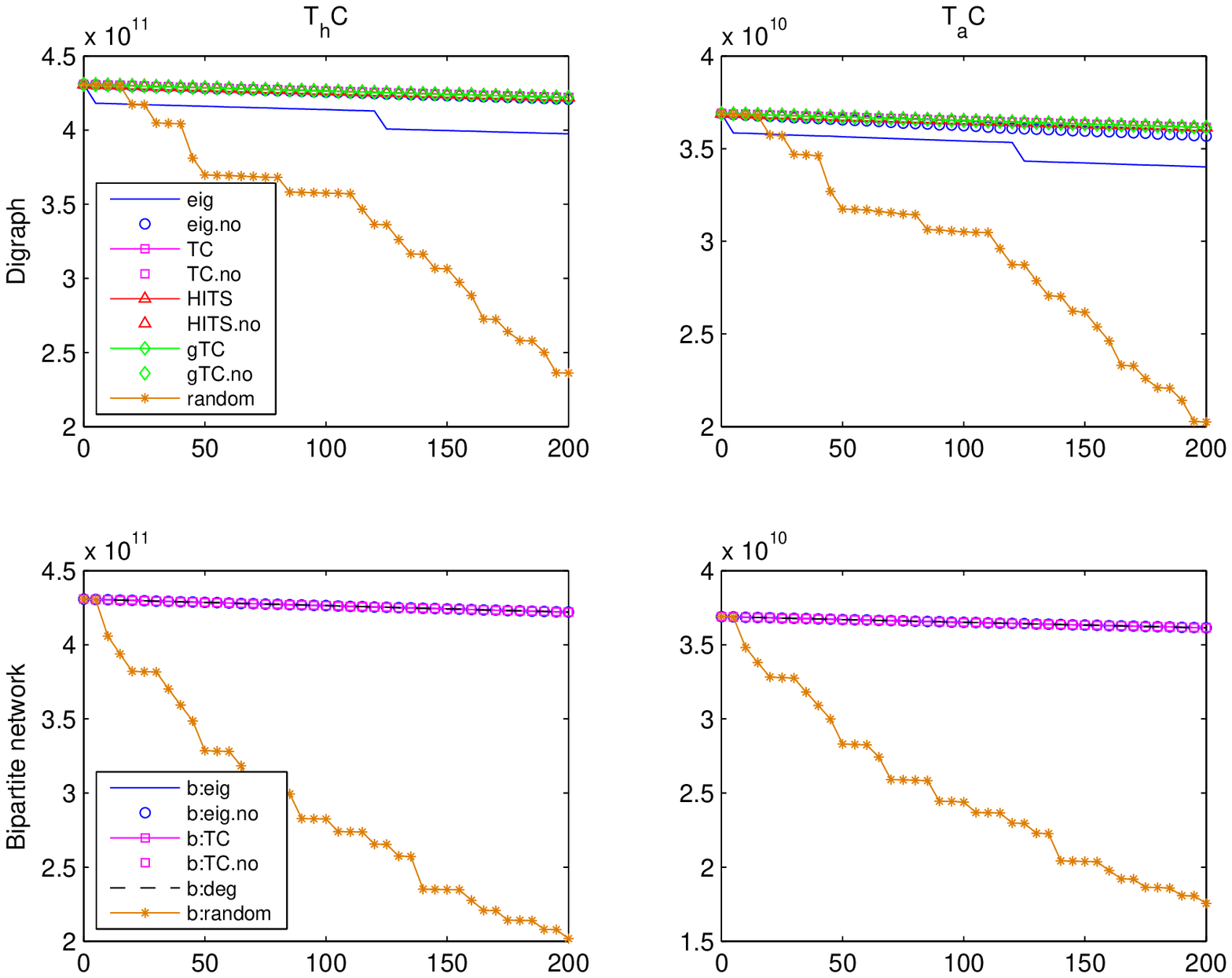}
\caption{Evolution of $T_hC$ and $T_aC$ for the network Abortion when 200 downdates are selected working on the 
matrix $A$ associated with the digraph (top) and on its bipartite version $\mathscr{A}$ (bottom).}
\label{fig:Abortion_down}
\end{figure}

\begin{figure}[t!]
\centering
\includegraphics[width=.83\textwidth]{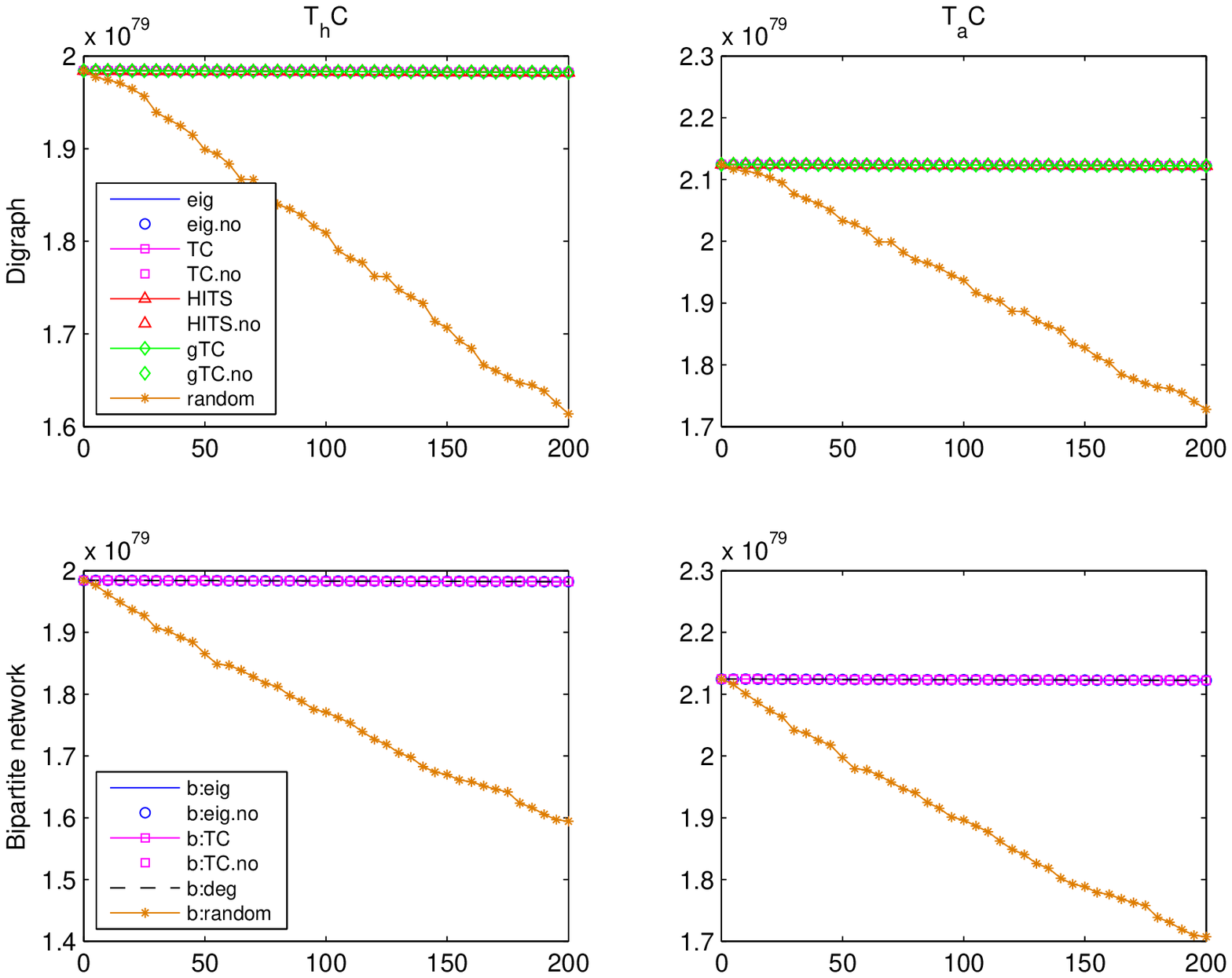}
\caption{Evolution of $T_hC$ and $T_aC$ for the network Twitter when 200 downdates are selected working on the 
matrix $A$ associated with the digraph (top) and on its bipartite version $\mathscr{A}$ (bottom).}
\label{fig:Twitter_down}
\end{figure}

\begin{figure}[h!]
\centering
\includegraphics[width=.83\textwidth]{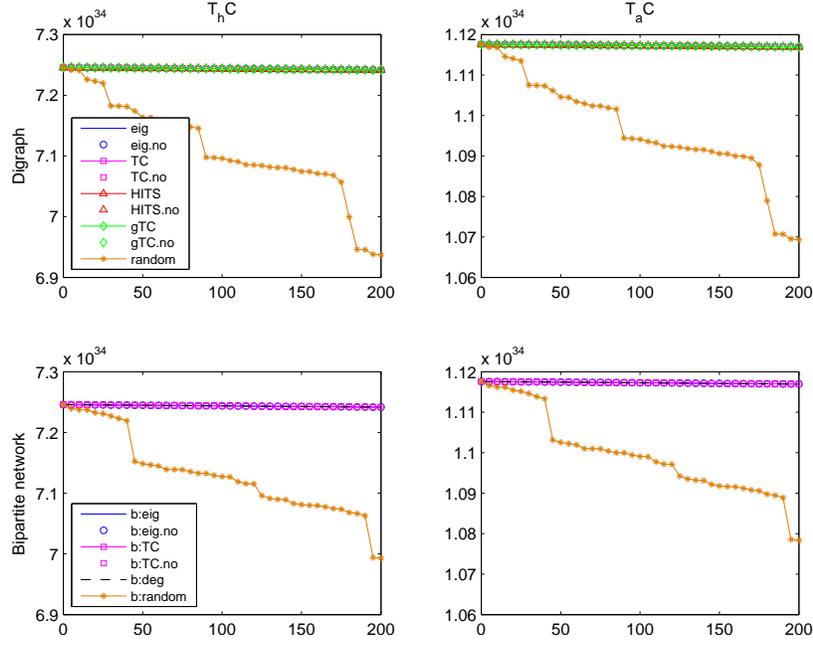}
\caption{Evolution of $T_hC$ and $T_aC$ for the network cit-HepTh when 200 downdates are selected working on the 
matrix $A$ associated with the digraph.}
\label{fig:cit-HepTh_down}
\vspace{0.1in}
\end{figure}

\begin{table}[t]
\footnotesize
\centering
\caption{Timings in seconds when $K=200$ downdates are selected for the networks in our Dataset using the methods described.}
\label{tab:time_Dwdt}
\begin{tabular}{lcccc}
\hline
                 &     Computational        &          &         &           \\
                 &      Complexity          & Abortion & Twitter & cit-HepTh \\
\hline
{\tt eig}        &        5.83              &    7.77  &   16.65 &   201.29  \\
{\tt eig.no}     &        0.04              &    0.04  &    0.05 &     1.03  \\
{\tt TC}         &      104.12              &   15.05  &   75.35 &   126.18  \\
{\tt TC.no}      &        0.92              &    0.07  &    0.39 &     0.73  \\
{\tt HITS}       &        2.72              &    4.71  &   13.32 &    63.40  \\
{\tt HITS.no}    &        0.02              &    0.02  &    0.08 &     0.34  \\
{\tt gTC}        &        5.49              &   12.06  &   60.77 &   175.02  \\
{\tt gTC.no}     &        0.04              &    0.08  &    0.29 &     0.80  \\
{\tt b:eig}      &        4.31              &    6.63  &   15.10 &    85.87  \\
{\tt b:eig.no}   &        0.03              &    0.04  &    0.05 &     1.89  \\
{\tt b:deg}      &        0.06              &    0.15  &    3.94 &     8.37  \\
{\tt b:TC}       &        5.51              &   11.66  &   39.02 &   126.44  \\
{\tt b:TC.no}    &        0.02              &    0.05  &    0.16 &     0.42  \\
\hline
\end{tabular}
\end{table}

Similar conclusions can be drawn when considering the results for the 
downdating problem, although the differences among the techniques 
are less pronounced (Figs.~\ref{fig:CC_down}--\ref{fig:cit-HepTh_down} and Table~\ref{tab:time_Dwdt}).
Indeed, the results shown 
confirm the effectiveness of the techniques based on the edge HITS and total communicability centralities and of their variants which 
do not require the recomputation of the rankings.
As in the case of the updating problem, the results returned by these two 
methods essentially reproduce those obtained when working on $\mathscr{A}$ 
using the heuristics {\tt b:eig(.no)} and {\tt b:TC(.no)}. 
 
The methods {\tt eig(.no)} and {\tt TC(.no)} perform no better
(and in some cases worse) than {\tt gTC(.no)} and {\tt HITS(.no)}, 
while {\tt b:deg} is usually outperformed by {\tt b:eig(.no)} or 
{\tt b:TC(.no)}. 

Concerning the timings, if we compare the results in Tables~\ref{tab:time_Updt} 
and~\ref{tab:time_Dwdt} we can see that the values in Table~\ref{tab:time_Updt} 
are in general higher that those in Table~\ref{tab:time_Dwdt}.  
This is easily understood in view of what we observed before, if one 
compares the number of virtual edges $\tau$ with the number of 
edges $m$ in each network in the dataset (see Table~\ref{tab:dataset}). 

While we do not provide a formal assessment of the computational cost of
the various heuristics, arguments similar to those found in \cite{AB14} indicate
that the cost of the more efficient heuristics can be expected to scale 
approximately like ${\pazocal O}(n)$ or ${\pazocal O}(n\log n)$ with the number of nodes $n$.

In conclusion, by considering the overall performance of the methods 
and their cost (in terms of timings), we find that
the best criteria for our updating/downdating goals are the methods 
{\tt HITS(.no)} and {\tt gTC(.no)}. Besides these, satisfactory
results may also be obtained using {\tt b:eig(.no)} or 
{\tt b:TC(.no)}. 
From the timings in Tables~\ref{tab:time_Updt} and~\ref{tab:time_Dwdt} we can deduce that the heuristics 
{\tt HITS(.no)} are in general slightly faster than {\tt b:eig(.no)} and may thus be preferred. 
Concerning whether it is better to use {\tt gTC(.no)} or {\tt b:TC(.no)}, we 
anticipate that the first will be preferrable when used in conjunction
with fast algorithms for the approximation of bilinear forms involving generalized matrix 
functions.

\section{Conclusions and future work}

In this work we have extended the notion of total network communicability
to directed graphs, and developed heuristics for manipulating an
existing directed network so as to enhance its communicability properties.
In doing so we made use of the concept of alternating walks,  
which allows us to take into account the dual role played by each node
in a digraph, namely, receiver and broadcaster of information. This
in turn led us in a natural way to the (rather overlooked)
concept of generalized matrix function, first introduced in \cite{HBI73}.
As shown in the paper, this concept
allows one to express various communicability
measures for digraphs in a compact form.

Our computational results indicate that the heuristics which take into account
the dual role of nodes in directed networks tend to be preferable to those
that do not. We also showed that these heuristics are very fast in
practice.

Future work will address computational issues for large-scale networks
(in particular, fast algorithms for estimating the row and column sums
of generalized matrix functions). Another avenue for future work is the
extension of the techniques in this paper and in \cite{AB14} to weighted
graphs. 
We also plan to investigate the use of our heuristics to tune 
other network properties. 
Preliminary tests suggest that our edge modification techniques are effective at 
increasing the synchronizability in directed graphs (see, e.g., \cite{ZLZ}).
A more systematic exploration of this application is left for future work.


\section*{Acknowledgements}
We are indebted to two anonymous referees for helpful suggestions.
The first author would like to thank Emory University for the hospitality 
offered in 2015, when part of this work was completed. 


\end{document}